%% file: two_calculations_arxiv.tex
\documentclass[toc=listof,toc=bibliography,10pt,table]
{scrartcl}
\makeatletter
\DeclareOldFontCommand{\rm}{\normalfont\rmfamily}{\mathrm}
\DeclareOldFontCommand{\sf}{\normalfont\sffamily}{\mathsf}
\DeclareOldFontCommand{\tt}{\normalfont\ttfamily}{\mathtt}
\DeclareOldFontCommand{\bf}{\normalfont\bfseries}{\mathbf}
\DeclareOldFontCommand{\it}{\normalfont\itshape}{\mathit}
\DeclareOldFontCommand{\sl}{\normalfont\slshape}{\@nomath\sl}
\DeclareOldFontCommand{\sc}{\normalfont\scshape}{\@nomath\sc}
\makeatother

\usepackage{amssymb}
\usepackage{latexsym}
\usepackage{amsmath}
\usepackage{enumerate}
\usepackage{amsthm}
\usepackage{wasysym}
\usepackage{stmaryrd}
\usepackage{mathrsfs}
\usepackage{pifont}
\usepackage{colortbl}
\usepackage[usenames,table]{xcolor}

\usepackage{tikz}
\usepackage{cases}

\makeatletter
\def\thm@space@setup{%
  \thm@preskip=\parskip \thm@postskip=0pt
}
\makeatother

\newcommand{\qee} {\hspace*{2mm}\hfill \ding{109}}

\renewcommand{\iff}{\leftrightarrow}
\renewcommand{\phi}{\varphi}

\definecolor{uuxgreen}{cmyk}{1,0,0.75,0}
\definecolor{uuxred}{cmyk}{0.2,1,0.9,0.1}
\definecolor{uuyblue}  {cmyk}{0.9,0.55,0,0}

\newtheorem{theorem}{Theorem}[section]
\newtheorem{define}[theorem]{Definition}
\newenvironment{definition}{\begin{define} \rm}{\qee\end{define}}
\newtheorem{exa}[theorem]{Example}

\newenvironment{exampleno}{\begin{exa} \rm}{\end{exa}}
\newtheorem{exerc}[theorem]{Exercise}

\newtheorem{conj}[theorem]{Conjecture}

\newtheorem{ques}[theorem]{Open Question}
\newenvironment{question}{\begin{ques} \rm}{\qee\end{ques}}
\newtheorem{lem}[theorem]{Lemma}
\newenvironment{lemma}{\begin{lem} \it}{\end{lem}}
\newtheorem{cor}[theorem]{Corollary}
\newenvironment{corollary}{\begin{cor} \it}{\end{cor}}
\newtheorem{rem}[theorem]{Remark}
\newenvironment{remark}{\begin{rem} \rm}{\qee\end{rem}}

\DeclareMathOperator{\possible}{\text{\tikz[scale=.6ex/1cm,baseline=-.6ex,rotate=45,line width=.1ex]{
                            \draw (-1,-1) rectangle (1,1);}}}
\DeclareMathOperator{\necessary}{\text{\tikz[scale=.6ex/1cm,baseline=-.6ex,line width=.1ex]{
                            \draw (-1,-1) rectangle (1,1);}}}

\DeclareMathOperator{\dotnecessary}{\text{\tikz[scale=.6ex/1cm,baseline=-.6ex,line width=.1ex]{
                            \draw (-1,-1) rectangle (1,1);  \draw[fill=black] (0,0) circle (.25);}}}

 \newcommand{\tupel}[1]{{\langle #1 \rangle}}
\newcommand{\verz}[1]{\{ #1 \}}
\newcommand{\To}{\Rightarrow}

\newcommand{\hyph}{\mbox{-}}

\newcommand{\gnum}[1]{{\ulcorner #1 \urcorner}}

\newcommand{\apr}{\vartriangle}

\newcommand{\opr}{\necessary}
\newcommand{\dotbox}{\dotnecessary}

\newcommand{\oco}{\possible}

\usepackage{wasysym}

\usepackage{proof}
\usepackage{parskip}
\usepackage{diagrams}

\usepackage{url}

\newcommand{\qedright}{\belowdisplayskip=-12pt }

\DeclareSymbolFont{extraup}{U}{zavm}{m}{n}
\DeclareMathSymbol{\vardiamond}{\mathalpha}{extraup}{87}
 \DeclareSymbolFont{symbolsC}{U}{txsyc}{m}{n}
\DeclareMathSymbol{\strictif}{\mathrel}{symbolsC}{74}
\DeclareMathSymbol{\strictfi}{\mathrel}{symbolsC}{75}
\DeclareMathSymbol{\strictiff}{\mathrel}{symbolsC}{76}
\newcommand{\tto}{\strictif}
\newcommand{\ifff}{\strictiff}

\renewcommand{\phi}{\varphi}

\definecolor{uuxgreen}{cmyk}{1,0,0.75,0}
\definecolor{uuxred}{cmyk}{0.2,1,0.9,0.1}
\definecolor{uuyblue}  {cmyk}{0.9,0.55,0,0}
\definecolor{airforceblue}{rgb}{0.36, 0.54, 0.66}
\definecolor{brickred}{rgb}{0.8, 0.25, 0.33}
\definecolor{ao}{rgb}{0.0, 0.0, 1.0}
\definecolor{cobalt}{rgb}{0.0, 0.28, 0.67}

\definecolor{light-gray}{gray}{0.92}

\newcommand{\quatb}{\ensuremath{4_{\opr}}}
\newcommand{\loebb}{\ensuremath{{\sf L}_{\opr}}}

\newcommand{\di}{{\sf Di}}
\newcommand{\na}{\ensuremath{\sf N_a}}
\newcommand{\ka}{\ensuremath{\sf K_a}}

\newcommand{\quata}{\ensuremath{4_{\sf a}}}
\newcommand{\quatao}{\ensuremath{4^\circ_{\sf a}}}
\newcommand{\qquatao}{\ensuremath{44^\circ_{\sf a}}}
\newcommand{\loeba}{\ensuremath{{\sf L_a}}}
\newcommand{\loebao}{\ensuremath{{\sf L^\circ_a}}}
\newcommand{\sloebra}{\ensuremath{{\sf sLR_a}}}
\newcommand{\tr}{{\sf Tr}}
\newcommand{\jv}{{\sf JV}}
\newcommand{\js}{{\sf JS}}
\newcommand{\weak}{{\sf W}}
\newcommand{\weako}{\ensuremath{{\sf W}^\circ}}
\newcommand{\weakst}{\ensuremath{{\sf W}^\ast}}
\newcommand{\pers}{{\sf P}}

\newcommand{\ipc}{{\sf IPC}}
\newcommand{\cpc}{{\sf CPC}}
\newcommand{\iam}{\ensuremath{{\sf iA}^{-}}}
\newcommand{\ia}{\ensuremath{{\sf iA}}}
\newcommand{\iglam}{\ensuremath{{\sf iGL}_{\sf a}^{-}}}
\newcommand{\iglamo}{\ensuremath{{\sf iGL}_{\sf a}^{\circ-}}}
\newcommand{\iglbm}{\ensuremath{{\sf iGL}_{\opr}^{-}}}
\newcommand{\iglpm}{\ensuremath{{\sf iGLP}^{-}}}
\newcommand{\iglwm}{\ensuremath{{\sf iGLW}^{-}}}
\newcommand{\iglwmo}{\ensuremath{{\sf iGLW}^{\circ-}}}

\newcommand{\igla}{\ensuremath{{\sf iGL}_{\sf a}}}

\newcommand{\iglb}{\ensuremath{{\sf iGL}_{\opr}}}
\newcommand{\iglp}{\ensuremath{{\sf iGLP}}}
\newcommand{\iglw}{\ensuremath{{\sf iGLW}}}

\newcommand{\app}[3]{(\lambda #1.#2)#3}

\newcommand{\stex}[2]{#2^{\lceil #1 \rceil}}
\newcommand{\thad}{\oplus}

\newcommand{\lna}[1]{\ensuremath{\mathsf{#1}}}

\usepackage[final]{fixme} 
\FXRegisterAuthor{av}{aav}{AV}
\FXRegisterAuthor{tmil}{atmil}{TL minor}
\FXRegisterAuthor{tl}{atl}{TL major}
\definecolor{mydarkgreen}{rgb}{0,0.34,0}
\newcommand{\tlnt}[1]{\tmilnote[inline,marginclue]{\textcolor{mydarkgreen}{#1}}}
\newcommand{\alnt}[1]{\avnote[inline,marginclue]{\textcolor{cobalt}{#1}}}

\newcommand{\takeout}[1]{}
\newcommand{\rfsc}[1]{\S\,\ref{#1}}
\newcommand{\rfs}[1]{\S\,#1}
\newcommand{\rfse}[1]{\S\,\ref{sec:#1}}

\newcommand{\ma}[1]{\mathfrak{#1}}
\newcommand{\deq}{:=}
\newcommand{\la}{\langle}
\newcommand{\ra}{\rangle}

\newcommand{\strictp}{\ensuremath{\strictif}-p\,}

\newcommand{\comp}{\!\cdot\!}
\newcommand{\kmodels}{\Vdash}

\usepackage{hyperref}
\newcommand{\lnref}[1]{

\newcommand{\subl}{{\sf a}}
\newcommand{\lb}{\lnref{\lna{Box}}}
\newcommand{\biii}{\lnref{4_{\opr}}}
\newcommand{\lv}{\lnref{4_{\subl}}}
\newcommand{\biv}{\lnref{{\sf L}_{\opr}}}
\newcommand{\lvii}{\lnref{{\sf W}}}
\newcommand{\lS}{\lnref{\lna{S}}}
\newcommand{\bk}{\lnref{\lna{{\mathrm i}\hyph BoxA}}}%
\newcommand{\ws}{\lna{{\mathrm i}\hyph SA}}
\newcommand{\loglb}{\lnref{{\sf {\mathrm i}A}}}
\newcommand{\loglg}{\lnref{{\mathrm i}\hyph{\sf GL}_{\subl}}}
\newcommand{\loglh}{\lnref{{\mathrm i}\hyph{\sf GW}}}

\usepackage[all]{xy}
\xyoption{2cell}
\xyoption{curve}
\UseTwocells
\SelectTips{cm}{}

\newcommand{\bro}{\textup{(}}
\newcommand{\brc}{\textup{)\,}}
\newcommand{\tuc}[2]{\textup{\cite[#1]{#2}}}

\newcommand{\nei}{\mathfrak{N}}

\newcommand{\ups}[1]{\mathtt{Up}^{\preceq}(#1)}
\newcommand{\pow}[1]{\mathtt{Pow}(#1)}

\newcommand{\sembr}[2]{[\!\![#1]\!\!]_{\ma #2}}

\newcommand{\semb}[1]{[\!\![#1]\!\!]}

\newcommand{\nk}[1]{\, \nei #1 \,}

\newcommand{\sqk}[1]{\sqsubset_{#1}}

\newcommand{\imag}[2]{{#1}_{#2}}
\newcommand{\velt}[1]{{\ma #1}^{\mathsf{ve}}}

\newcommand{\gA}{\mathfrak{A}}
\newcommand{\gB}{\mathfrak{B}}
\newcommand{\gC}{\mathfrak{C}}

\newcommand{\gF}{\mathfrak{F}}
\newcommand{\gG}{\mathfrak{G}}

\newcommand{\bao}{\ensuremath{\textsc{bao}}}
\newcommand{\baos}{\ensuremath{\textsc{bao}}s}

\newcommand{\hae}{\ensuremath{\textsc{hae}}}
\newcommand{\haes}{\ensuremath{\textsc{hae}}\textup{s}}

\newcommand{\thaes}{\ensuremath{\tto\!\hyph\textsc{hae}}\textup{s}}
\newcommand{\dihaes}{\ensuremath{\di\hyph\textsc{hae}}\textup{s}}
\newcommand{\stp}[2]{#2^{[ #1 ]}}

\newcommand{\sfl}{\lna{S4Lew}} 
\newcommand{\wzbm}{\lna{BM}} 

\newcommand{\pf}[1]{PF(#1)}
\newcommand{\adm}[1]{P_{#1}}

\newcommand{\eqc}[1]{[#1]}

\newcommand{\ha}{\ensuremath{\mathsf{HA}}}

\newcommand{\logt}[1]{\Lambda_F(#1)}
\newcommand{\logpt}[1]{\Lambda^{\circ}(#1)}
\newcommand{\logptd}[1]{\Lambda^{\circ}_{\Delta}(#1)}

\newcommand{\lpl}{\mathbf{\oplus}}

\newcommand{\lgax}[1]{\Lambda^-(#1)}

\newcommand{\refeq}[1]{\ensuremath{\textup{(\ref{eq:#1})}}}
\newcommand{\refe}[1]{\ensuremath{\textup{(\ref{#1})}}}

\newcommand{\Mod}[1]{\mathsf{Mod}(#1)}
\newcommand{\Th}[1]{\mathsf{Th}(#1)}

\usepackage{catchfilebetweentags}  

\makeatletter
\newif\if@conf
 \@conffalse   
 \ProcessOptions*\relax
\let\ifconf\if@conf

\newcommand{\confbl}{See appendix/technical report.}
\newcommand{\neskip}{\ifconf\else\newline\fi}

\usepackage{listings}

\lstdefinelanguage{mace}{morekeywords={formulas,assumptions,goals,end_of_list,
  PROOF, end, THEOREM, PROVED, of,statistics,length,level,search},
  moredelim=**[s][\rmfamily\itshape]{[}{]},
  basicstyle=\ttfamily,
  morekeywords= [2]{para,
  rewrite, flip, back_rewrite, copy, deny,resolve,assumption,goal,label,non_clause}, 
  sensitive=false}
\lstnewenvironment{macecode}
  {\lstset{keywordstyle={\color{black}\sffamily\bfseries},keywordstyle={[2]\itshape\sffamily},keywordstyle={[3]\color{green}\sffamily},language={mace},breaklines=true,columns=fullflexible}}
  {}

\usepackage{multirow}
\usepackage{tabularx}
\newcolumntype{L}[1]{>{\raggedright\let\newline\\\arraybackslash\hspace{0pt}}m{#1}}
\newcolumntype{C}[1]{>{\centering\let\newline\\\arraybackslash\hspace{0pt}}m{#1}}
\newcolumntype{R}[1]{>{\raggedleft\let\newline\\\arraybackslash\hspace{0pt}}m{#1}}

\title{Lewisian Fixed Points I:  \\ Two Incomparable Constructions}
  \author{Tadeusz Litak \and Albert Visser}
 \titlehead{Informatik 8, FAU Erlangen-N\"{u}rnberg,  Germany \url{tadeusz.litak@fau.de} \newline
 
     Philosophy,  Faculty of Humanities, 
                Utrecht University,  \newline
               Janskerkhof 13,
                3512BL~~Utrecht,  The Netherlands  \url{a.visser@uu.nl}}

\date{}

\begin{document}
\maketitle


\begin{abstract}
Our paper is the first study of what one might call ``reverse mathematics of explicit fixpoints''. We study two methods of constructing such fixpoints for formulas whose principal connective is the intuitionistic 
Lewis arrow $\tto$. Our main motivation comes from metatheory of constructive arithmetic, but the systems in question allows several natural semantics. 
  The first of these methods, inspired by de Jongh and Visser, 
turns out to yield a well-understood modal system \iglam.  The second one by de Jongh and Sambin, seemingly simpler, leads to a 
modal theory  $\iam\thad\js$, which proves harder to axiomatize in an elegant way. Apart from showing that both theories 
are incomparable, we axiomatize their join and investigate several subtheories, whose axioms are obtained as fixpoints of simple formulas. 
 We also show that both \iglam and  $\iam\thad\js$ are \emph{extension stable}, that is, their validity in the corresponding preservativity  logic of a given arithmetical theory transfer to its finite extensions. 
\end{abstract} 

\tableofcontents

\section{Introduction}

\emph{Provability logic} studies propositional and algebraic aspects of arithmetical theories,  their provability predicates and reflection principles. 
 Thanks to Solovay's arithmetical completeness result \cite{solo:prov76}, we know that the provability predicate of Peano Arithmetic
 \cite{smor:self85,bool:emer91,Boolos1993,lind:prov96,japa:logi98,svej:prov00,arte:prov04,halb:henk14}    yields precisely the famous system \lna{GL}, also known as the (\emph{G\"odel}-)\emph{L\"ob} logic, obtained from the minimal unimodal normal logic by adding the principle $\opr(\opr\phi \to \phi) \to \opr\phi$.  One of the most important facts about \lna{GL}\ is that it allows definability of explicit fixpoints. That is, given any polynomial $\phi(p)$ where all occurrences of $p$ are guarded by $\opr\,$, one can use de Jongh-Sambin algorithm to compute a formula $\chi$ not involving $p$ and $\lna{GL} \vdash \chi \leftrightarrow \phi(\chi)$; \tlnt{we should actually state what $\vdash$ stands for!}
 furthermore, $\chi$ thus computed is unique up to propositional equivalence (de Jongh, Sambin  \cite{samb:effe76}, Bernardi \cite{bern:uniq76}).  Actually, \lna{GL}\ is obtained precisely as the smallest extension of \lna{K4} (i.e.,
 the logic of the transitivity axiom $\opr \phi \to \opr\opr \phi$) in which guarded fixpoints are definable. 
 This follows immediately from the fact that, by L\"ob's argument, L\"ob's principle is entailed by the presence of
 guarded fixed points in combination with the de Jongh-Sambin result.

This result encodes the algebraic content of the L\"ob Theorem and G\"odel's Second Incompleteness Theorem. The modal analysis gives us the conceptual resources 
 to say that the consistency statement
 is the explicit form of the G\"odel sentence.  
  More mundanely, it can be seen as elimination of fixpoint operator.  The original statement is restricted to guarded fixpoints, but one can indeed extend this classical result to  elimination of positive  fixpoints of ordinary $\mu$-calculus  \cite{bent:moda06,viss:lobs05} and further beyond (see \rfsc{sec:conclusions}).  Given the Kripkean meaning of \lna{GL}\ as the logic of Noetherian (conversely well-founded) transitive frames, such results in turn have found applications, e.g., in characterizing expressivity of XPath fragments \cite[\rfs{3.1}]{CateFL10:jancl}. The Sambin-de Jongh result has inspired Nakano's seminal work on \emph{modality for recursion} \cite[\rfs{7}]{Nakano00:lics}, \cite[\rfs{5}]{Nakano01:tacs}. Last, but definitely not the least, it can be used to prove other metaresults about \lna{GL}, such as the Beth definability property, as observed first by Maximova \cite{Maksimova1989,ArecesHJ98,hoog:defi01,iemh:prop05}   (cf. \rfsc{sec:beth}).


What happens when we broaden the investigation beyond  the classical base and unary provability $\opr$\,? Regarding the former restriction, already Sambin's 1976 paper \cite{samb:effe76} noted that the fixpoint theorem works over intuitionistic propositional calculus (\ipc). Unfortunately, despite decades of efforts  \cite{viss:prop94,Iemhoff01:phd,iemh:moda01,viss:close08,arde:sigm14}, there is no known axiomatization of the provability logic of Heyting Arithmetic (\lna{HA}) and related systems; it is a system much stronger than intuitionistic \lna{GL}, including principles such as \[\opr(\neg\neg\, \opr \phi \to \opr\phi)\to \opr\opr\phi \] 
underivable even in classical \lna{GL}, as classically it implies 
 $\opr\opr\bot$ (cf. \cite[\S 5.3]{LitakV18:im} for more examples). 
 The algebraic core of a weak theory can include powerful schemes refutable in a stronger theory enjoying a different provability predicate. This phenomenon is often caused, e.g., by the fact 
   that the weaker theory is closed under some translation method, whereas the stronger one is not.  

Allowing non-unary connectives opens up vast new landscapes, especially in the constructive setting. In our paper \cite{LitakV18:im}, apart from providing a general framework of \emph{schematic logics} \cite[\rfs{5.1}]{LitakV18:im}, we have made the case for the \emph{constructive strict implication} $\tto$, also called the \emph{Lewis arrow}.  It allows defining $\opr\phi$  as $\top \tto \phi$. We list its arithmetical interpretations in \rfsc{sec:arint}. The most important one  in the study of metatheory of {\sf HA} is provided by $\Delta$-preservativity for a theory $T$, where $\Delta$ is a class of 
 sentences, most commonly taken to be $\Sigma^0_1$  \cite{viss:eval85,viss:prop94,viss:subs02,iemh:pres03,iemh:prop05}.  Other ones include contraposed conservativity/interpretability \cite{bera:inte90,shav:rela88,japa:logi98,viss:over98,arte:prov04} \cite[\rfs{C.3}]{LitakV18:im},  the logic of 
admissible schemes or the $\boxdot$-provability interpretation. See \rfsc{sec:arint} for details. 
Almost needless to say, constructive arithmetic can be replaced by any other foundational theory rich enough for standard encodings of syntactic notions. There are also computational interpretations originating elsewhere, such as Hughes ``classical'' arrows in functional programming  \cite{Hughes00:scp,LindleyWY08:msfp} (cf. \cite[\rfs{7.1}]{LitakV18:im}).  All these interpretations extend the base system $\iam$ introduced in \rfsc{sec:logics}. Interestingly enough, not all of them validate principles like $\di$ (cf. \rfsc{sec:logics}) which hold in the standard Kripke semantics of $\tto$ (cf. \cite[\rfs{3}]{LitakV18:im} and \rfsc{sec:kripke}). \takeout{ hence one often resorts to 
 purely syntactic derivations or more liberal semantics.} 
   But what axioms do we need to ensure fixpoint results, and how can we compute these explicit fixpoints?

\subsection{Our Contributions} 


\tlnt{We still need to credit IDZ properly.}


The classical construction of explicit guarded fixed points proceeds in two stages. One first proves the result for formulas where the main
connective is the modal operator and then one shows how to extend the result to all modalized formulas (and possibly beyond that).
Our results directly concern the first step. See \rfsc{sec:conclusions}  for a brief discussion of the second one. 

We have two known paradigms for such a construction. 
 First, there is the original de Jongh-Sambin construction (see \rfse{js}) as generalized by 
  Smory\'nski \cite[Ch. 4]{smor:self85}. Secondly, there is the construction given by de Jongh and Visser \cite{dejo:expl91}
for the interpretability logic {\sf IL}  (see \rfsc{iglamjv}). 
  As it simplifies to the de Jongh-Sambin
construction when one adds the principle {\sf W} (see Figure 
\ref{tab:compl}), it seemed the \emph{master} construction.
 Our results show that de Jongh-Visser construction and the de Jongh-Sambin construction are
mutually incomparable. The incomparability result also holds in the classical case. 

Our paper is the first study of what one might call ``reverse mathematic of explicit fixpoints''. After discussing  algebraic and Kripke semantics for extensions of $\iam$ (\rfse{semantics}), we investigate the effect of adding explicit schemes stating that a given method (de Jongh-Visser or de Jongh-Sambin) indeed yields fixpoints of formulas whose principal connective is $\tto$. 
 This, however, requires a significant prerequisite: we noted above that the validity of a scheme in the logic (algebraic core) 
 of a given arithmetical theory (\rfse{arint}) does not need to transfer  to the logic of some given finite extension. As we show in a companion paper \cite{tlav19subf}, this holds if the base logic enjoys the property of \emph{extension stability}. 
After recalling this information (\rfsc{sec:exsta}), we show that the minimal theory in which the de Jongh-Visser construction works is the theory \iglam (\rfsc{iglamjv}), which is extension stable. Thus,
for the de Jongh-Visser construction we have a precise analogue of L\"ob's Logic.  
In \rfsc{sec:js} we show that the case of the de Jongh-Sambin
construction is more complex 
 and show the incomparability of its theory ($\iam\thad\js$, which also turns out to be extension stable) with \iglam. In \rfse{join} we axiomatize the join of both theories.  In \rfsc{sec:jscloser} we investigate several subtheories of $\iam\thad\js$, whose axioms are obtained as de Jongh-Sambin fixpoints of simple formulas. A large part of our results on axiomatizing explicit fixpoints is concisely summarized  by Figure \ref{fig:logics} therein. 
 In \rfse{corresp} we present Kripke semantics for some principles investigated in earlier sections   and uncover a simple nonconservativity phenomenon.

\section{Basics}  \label{sec:basics}  \label{sec:logics}
%

Our basic system is \iam in  the language of intuitionistic propositional calculus ({\ipc})  extended with a binary 
connective $\tto$. 
We write $\opr\phi$ for $\top\tto \phi$.
 The system  is given by the following axioms:
\begin{description}
\item[prop]
axioms and rules for {\ipc}
\item[\tr]
$((\phi \tto \psi) \wedge (\psi \tto \chi)) \to (\phi \tto \chi)$
\item[\ka]
$((\phi \tto \psi) \wedge (\phi \tto \chi)) \to (\phi \tto (\psi\wedge\chi))$
\item[\na]
$\vdash \phi \to \psi \;\;\; \To \;\;\; \vdash \phi \tto \psi$ 
\end{description}

We take $\opr\phi := \top \tto \phi$ and $\dotbox\phi$ for $\phi\wedge \opr\phi$.
One can easily derive the  intuitionistic version  of the classical system {\sf K} (without $\Diamond$) 
 for the $\opr$-language
from \iam. 
The system \ia\ extends \iam\ with  
\begin{description}
\item[\di]
$((\phi \tto \chi) \wedge (\psi\tto\chi)) \to ((\phi\vee\psi)\tto\chi)$
\end{description}


A (\iam-)\emph{logic} $\Lambda$ is an extension of \iam\ that is closed under modus ponens, necessitation and substitution.
Let $X$ be a set of formulas. We write $\Lambda \thad X$ for the closure of $\Lambda \cup X$ under modus ponens and necessitation.
Note that $\Lambda \thad X$ is not automatically a logic. 
 On the other hand, if $X$ is closed under substitution, then so
is $\Lambda \thad X$. 

\begin{remark}
All  theorems we claim for \iam\ also hold  when  we omit disjunction from the language, in the sense that
we still have all schemes, where 
 the interpretations of the schematic letters are restricted to disjunction-free formulas. Our proofs also work in the disjunction-free setting. 

Theorems~\ref{ququatao} and \ref{iamququatao} illustrate that $\di$ is sometimes needed to derive principles not involving $\vee$. 
A similar example is provided by the trivialization of $\tto$ in $\ia\thad\cpc$ 
 \cite[Lemma 4.6]{LitakV18:im}.
 In the latter case, we know we need \di\ to make the argument work since the classical interpretability logic {\sf IL} does not trivialize.
\end{remark}

At some points,  we will use a convenient notation for substitution.
Suppose a variable, say $r$, of substitution is given in the context. We will write
 $\phi\psi$ for $\phi [r:=\psi]$. We note that $(\phi\psi)\chi$ is equal to $\phi(\psi\chi)$.
 So we may write $\phi\psi\chi$.


\begin{lemma}\label{sub1sub2}
Let a designated variable of substitution $r$ be given.
We have:
\begin{description}
\item[\sf Sub1]
$\iam\thad \quatb \vdash \dotbox(\phi \iff \psi) \to (\chi\phi \iff \chi\psi)$.
\item[\sf Sub2]
Suppose every occurrence of $r$ is in the scope of an occurrence $\tto$ in $\chi$. We have:\\
$\iam\thad \quatb \vdash \opr(\phi \iff \psi) \to (\chi\phi \iff \chi\psi)$.
\end{description}
\end{lemma}

\begin{proof}
\ifconf\confbl\else
\ExecuteMetaData[apptwocalc.tex]{sub1sub2}
\fi
\end{proof}

\begin{lemma}\textup{(Uniqueness of fixpoints)}\label{brilsmurf}
Fix a designated variable of substitution $r$. 
 Suppose every occurrence of $r$ is in the scope of an occurrence $\tto$ in $\chi$. Suppose $q$ does not occur in $\chi$. We have:
\begin{enumerate}[a.]
\item
$\iam\thad \loebb \vdash (\dotbox (r\iff \chi r) \wedge \dotbox (q\iff \chi q)) \to (r \iff q)$.
\item
$\iam\thad \loebb \vdash (\dotbox (r\iff \chi r) \wedge \dotbox (q\iff \chi q)) \to \dotbox(r \iff q)$.
\item
$\iam\thad \loebb \vdash (\opr (r\iff \chi r) \wedge \opr (q\iff \chi q)) \to \opr(r \iff q)$.
\end{enumerate}
\end{lemma}

\begin{proof}
\ifconf\confbl\else
\ExecuteMetaData[apptwocalc.tex]{brilsmurf}
\fi
\end{proof}

\begin{remark}
We find a closely related development in Craig Smory\'nski's \cite[Chapter 4]{smor:self85}.
We note that our development is constructive and studies a binary operator, where Smory\'nski works classically and studies a unary
operator. 
Moreover, since Smory\'nski aims at a version of the de Jongh-Sambin Theorem, he imposes an additional   axiom written (in an adjusted notation) as
$\vdash {\apr}\phi \to \opr{\apr}\phi$.
\end{remark}

\section{Semantics} \label{sec:semantics} 

Most proofs in this paper are of purely syntactic nature. Nevertheless, we occasionally  still need semantics, e.g., for non-derivability results and for better understanding of syntactic systems and notions studied below. For our purposes, it is enough to consider algebraic semantics (\rfsc{sec:algebras}) 
 and Kripke semantics  (\rfsc{sec:kripke}). 

\subsection{Algebraic Semantics} \label{sec:algebras}

We briefly recapitulate algebraic semantics as discussed in our companion paper \cite{tlav19subf}. A generic  algebraic completeness result after the manner of Lindenbaum and Tarski is obtainable for almost any ``natural'' logic and extensions of $\iam$ are no exception. We can in fact put it in a general setting: with their 
The ``global consequence relation'' of any extension of $\iam$ including both  Modus Ponens and $\na$ is easily seen to be  \emph{algebraizable} \cite{BlokP89:ams,FontJP03a:sl},  in fact an instance of what Rasiowa calls an \emph{implicative logic}  \cite{Rasiowa74:aatnl,Font06:sl}. 
More explicitly, the algebraic semantics looks as follows:

\begin{definition}
A \emph{$\tto$-algebra}\tlnt{a good name?} or \emph{\iam-algebra} is a tuple  
$\gA \deq \la A, \wedge, \vee, \tto, \to, \bot, \top \ra$, 
where 
\begin{itemize}
\item $\la A, \wedge, \vee, \to, \bot, \top \ra$ (the \emph{intuitionistic/Heyting reduct} of $\gA$) is a Heyting algebra and
\item $\la A, \wedge, \vee, \tto, \bot, \top \ra$ (the \emph{strict reduct} of $\gA$) satisfies
\begin{description}
\item[\lna{CK}] $(a \tto b) \wedge (a \tto c) = a \tto (b \wedge c)$,
\item[\lna{CT}] $(a \tto b) \wedge (b \tto c) \leq a \tto c$,
\item[\lna{CI}] $a \tto a = \top$.
\end{description}
\end{itemize}
Moreover, $\gA$ is a \emph{normalized} \tlnt{other names? ``fully normal''? ``additive''? ``$\tto$-Kripke''?}  or \emph{\ia-algebra} if its strict reduct is a weakly Heyting algebra \cite{CelaniJ05:mlq}, i.e., it satisfies in addition
\begin{description}
\item[\lna{CD}] $(a \tto c) \wedge (b \tto c) = (a \vee b) \tto c$.
\end{description}
\end{definition}

  \lna{CK}, \lna{CD}, \lna{CT},  and \lna{CI} are called by Celani and Jansana \cite{CelaniJ05:mlq} \lna{C1}\,--\,\lna{C4}, respectively. Denote the equational class of $\tto$-algebras by $\thaes$ (\haes\ standing for ``Heyting Algebra Expansions'') and the class of normalized ones by \dihaes. 

\tlnt{Also note that normalized ones are instances of setting by Palmigiano et al?}

A valuation $v$ in $\gA$ as usual maps propositional atoms to elements of $A$ and is inductively extended to $\hat{v}$ defined on all formulas in the obvious way. Write $\gA, v \Vdash  \phi$ if $\hat{v}(\phi) = \top$, $\gA \Vdash \phi$ if $\gA, v \Vdash \phi$ for all $\phi$, $\gA \Vdash \Lambda$ if $\gA \Vdash \phi$ for every $\phi \in \Lambda$, and $K \Vdash \phi$ if $\gA \Vdash \phi$ for every $\gA \in K$. Given any $K \subseteq \thaes$ and any set of formulas $\Lambda$ define
\begin{align*}
\Mod{\Lambda} & \deq \{ \gA \mid \forall \phi \in \Lambda. \gA \Vdash \phi\}, \\
\Th{K} & \deq \{ \phi \mid \forall \gA \in K. \gA \Vdash \phi \}.
\end{align*}

\begin{theorem}[Algebraic Completeness]\ \label{algsound}
\begin{itemize}
\item For any $K \subseteq \thaes$, $\Th{K}$ is a logic.
\item For any logic $\Lambda$, $\Lambda = \Th{\Mod{\Lambda}}$.
\end{itemize}
\end{theorem}

\begin{proof}
See our companion paper \cite{tlav19subf} or apply techniques of abstract algebraic logic (AAL) as discussed in standard references \cite{BlokP89:ams,FontJP03a:sl,Rasiowa74:aatnl,Font06:sl}.
\end{proof}

%


\subsection{Kripke Semantics} \label{sec:kripke}

 We briefly recapitulate basic information from our overview \cite{LitakV18:im}. 
 A (\emph{Lewisian}) \emph{Kripke frame} is a triple  $\ma F \deq \la W, \preceq, \sqsubset \ra$, where $\preceq$ is a partial order on $W$, $\sqsubset \; \subseteq \, W \times W$ and furthermore
\begin{description}
\item[\strictp] if $k \preceq \ell \sqsubset m$, then $k \sqsubset m$ \qquad (i.e., $\preceq \comp \;\sqsubset\; \subseteq\;  \sqsubset$).   
\end{description}
Admissible valuations then map propositional atoms to $\preceq$-closed sets and intuitionistic connectives are interpreted as usual using $\preceq$. The clause for $\tto$ in a model is 
\begin{equation} \label{eq:forctto}
 k \kmodels \phi \tto \psi \text{ if, for all } \ell \sqsupset k \text{, if } \ell \kmodels \phi \text{, then } \ell \kmodels \psi.
\end{equation}

Given such a Kripke frame, it is an easy exercise to define its dual algebra whose Heyting reduct is given as usual by $\ups{W}$, i.e., upsets of $W$, and the strict reduct (interpretation of $\tto$) is induced by \refeq{forctto}. The dual algebra is normalized, i.e.,  the equality \lna{CD} corresponding to \di\ holds. In other words, extensions of $\iam$ which are not extensions of $\ia$ can only be Kripke sound, but not Kripke complete. A fuller discussion of \di\ and other principles in the Kripke setting (written without employing explicitly algebraic language) can be found in our paper \cite{LitakV18:im}. Definitions of notions such as the  \emph{finite model property} (fmp), i.e., completeness wrt a finite class of frames are standard.

Some modal axioms that we need together with their correspondence conditions are given in Figure \ref{tab:compl}. Its fuller version can be found in \cite[\rfs{6}]{LitakV18:im}, along with an extended version of the following summary of existing results.

\begin{theorem}\ \label{th:minco} 
\begin{enumerate}[a.]
\item 
$\ia$ 
 has the finite model property  \tuc{Prop. 4.1.1}{Iemhoff01:phd}, \tuc{Prop. 7}{iemh:pres03}, \tuc{Th. 2.1.10}{Zhou03}.  
\item $\bk = \ia + \lb$ \bro$\ws  = \loglb + \lS$, $\loglb + \biii$\brc correspond to the class of brilliant frames (strong frames, semi-transitive frames)  and enjoy the fmp. 
\item $\loglb + \lv$ corresponds to the class of gathering frames and has the fmp 
\tuc{Prop. 4.2.1}{Iemhoff01:phd}, \tuc{Prop. 8}{iemh:pres03}.
\item $\loglb + \biv$ corresponds to the class of Noetherian semi-transitive frames and has the fmp 
 \tuc{Prop. 4.3.2}{Iemhoff01:phd}, \tuc{Th. 2.2.7}{Zhou03}.
\item \label{compgla} $\loglg  = \loglb + \biv + \lv$  corresponds to the class of Noetherian gathering frames \tuc{Lem. 9}{iemh:pres03}, \tuc{Lem. 3.10}{iemh:prop05} and has the fmp. 
\item \label{compgwa} $\loglh = \loglb + \lvii$ corresponds to the ``supergathering'' property of Figure \ref{tab:compl} 
 on the class of finite frames \tuc{Lem. 3.5.1}{Zhou03}, \tuc{Th. 3.31}{iemh:prop05}.
\end{enumerate}
\end{theorem}


\begin{question}
 Are $\loglh$ and related systems involving $\lvii$ Kripke complete?
 \end{question}

\tlnt{Mention SL too?} 
\newcommand{\sepo}{\vspace{\tbskip}}
\def\tbskip{2mm}

\begin{figure}
\hrule
\vspace{0.2cm}
\footnotesize

\begin{tabular}{lL{3.5cm}L{4.5cm}c} \sepo
\lb &  $\phi  \strictif \psi \to 
 \opr(\phi  \to \psi)$  \newline brilliant &  $k \sqsubset \ell \preceq m$ $\To$ $k \sqsubset m$ & 
$   \vcenter{
    \xymatrix@-1pc{
            & m\\
           k \ar@{~~>}[ur] 
           \ar@{~>}[r] & \ell \ar[u]      
      } }$ 
\\ [\tbskip]\sepo 
 \biii & $ \opr \phi \to \opr\opr\phi$  \newline semi-transitive & $k \sqsubset \ell \sqsubset m$ $\To$ \newline
 $\exists x. k \sqsubset x \preceq m$ 
&
$    \vcenter{
        \xymatrix@-0.5pc{
           x \ar@{-->}[r] & m\\
           k  \ar@{~~>}[u]
           \ar@{~>}[r] & \ell  \ar@{~>}[u]      
        }
      } $ 
 \\ [\tbskip]\sepo 
\lv & $ \phi \tto \opr \phi$ \newline gathering & $k \sqsubset \ell \sqsubset m$ $\To$ $\ell \preceq m$ 
&
$    \vcenter{
        \xymatrix{
          k  \ar@{~>}[r]  & \ell \ar@/_/@{-->}[r] \ar@/^/@{~>}[r]  & m \\
        }
      } $ 
\\ [\tbskip]\sepo 
\biv & $\opr(\opr\phi \to \phi) \to 
 \opr\phi$ & \multicolumn{2}{L{7.5cm}}
 {L\"ob: $\sqsubset$-Noetherian (conversely well-founded) and semi-transitive} \\ [\tbskip]\sepo 
\lvii & $ ((\phi \wedge \opr \psi) \tto \psi) \to 
  \phi \tto \psi$ \newline supergathering & on finite frames: \newline
$k\sqsubset \ell \sqsubset m \To 
 \exists x \sqsupset k. (\ell \prec x\preceq m)$  &
$    \vcenter{
        \xymatrix@-1pc{
           & x \ar@{-->}[r] & m\\
           k \ar@{~~>}[ur] 
           \ar@{~>}[rr] & & \ell \ar@{-->}|{\neq}[ul]  \ar@{~>}[u]      
        }
      } $ 
 \\ [\tbskip]\sepo      
\lS &  $\phi \to \opr\phi$  \newline strong & $k \sqsubset \ell$ $\To$ $k \preceq \ell$ &
$      \vcenter{
        \xymatrix{
          k  \ar@/_/@{-->}[r] \ar@/^/@{~>}[r] & \ell
        }
      } $
\end{tabular}
\caption{\label{tab:compl}Basic axioms and their correspondence conditions \cite[\rfs{6}]{LitakV18:im}. Further correspondence conditions are provided in \rfse{corresp}.}
\vspace{0.1cm}
\hrule
\end{figure}

\begin{remark}
There are other semantics for $\iam$ and $\ia$, which are intermediate between algebraic and Kripke ones. We do not treat them here, as our paper is primarily syntactic in nature, but in subsequent work we are going to discuss, e.g., suitable generalizations of so-called Veltman semantics, or semantics merging Kripkean interpretation of intuitionistic connectives with nieghbourhood interpretation of $\tto$. 
\end{remark}

\takeout{
****************
\begin{definition}[Stone-J\'{o}nsson-Tarski dual]
Given a normalized  $\tto$-algebra ($\ia$-algebra, i.e., algebra modelling \lna{Di}) $\gA$, define its \emph{dual descriptive frame} as $\gA_* \deq (\pf{\gA}, \preceq_\gA, \sqsubset_\gA, \adm{\gA})$, where 
\begin{itemize}
\item $\pf{\gA}$ is the collection of prime filters of $\gA$,
\item $\preceq_\gA$ is the inclusion relation between the prime filters,
\item $U \sqsubset_\gA V$ holds if for any $a, b \in A$, we have that $b \in V$ whenever $a \tto b \in U$ and $a \in V$,
\item and $\adm{\gA} \deq \{ a_* \mid a \in A\}$, where $a_* \deq \{U \in \pf{\gA} \mid a \in U\}$. 
\end{itemize}
\end{definition}

\begin{itemize}
\item 
For any Kripke frame $\gF$, $\gF^+$ is a normalized  $\tto$-algebra.
\item For any general frame $\gF$, $\gF^*$ is a normalized  $\tto$-algebra.
\item For any normalized  $\tto$-algebra $\gA$, $\gA_*$ is a general frame. Moreover, $\gA$ is isomorphic to $(\gA_*)^*$. 
\end{itemize}
\end{theorem}
}

\section{Arithmetical Interpretations} \label{sec:arint}


Our paper \cite{LitakV18:im} proposed a general framework of \emph{schematic logics} \cite[\rfs{5.1}]{LitakV18:im} for arithmetical interpretations of logical systems, in particular extensions of $\iam$. It can be also seen as semantics of propositional modal logics and this aspect is
 our main focus here. Let us present the framework restricted to the $\tto$-signature. 
 
 We restrict ourselves in our presentation to  theories extending the intuitionistic version of Elementary Arithmetic in the same language.
 Intuitionistic Elementary Arithmetic consists of the basic axioms for successor, addition and muliplication plus
 $\Delta_0$-induction, plus an axiom expressing the exponentiation is total over intuitionistic predicate logic.  
 We assume that a $\Sigma^0_1$-formula $\sigma_T$ representing the given axiom set of $T$ is part of the data for $T$. NB: we thus treat
 arithmetical theories differently from propositional logics that are given as sets of sentences.  
 
 Let a function $F$ that assigns to $T$ an arithmetical formula $A_T(v_0,v_1)$ as interpretation of $\tto$. 
 We note that $F$ operates intensionally on $\sigma_T$.
 We write $B_0 \tto_{F,T} B_1$ for $F(T)(\gnum{B_0},\gnum{B_1})$. Here $\gnum{C}$ is the numeral of the
G\"odel number of $C$.
Suppose $f$ is a mapping from the propositional atoms to
arithmetical sentences. 
 We define $(\phi)_{F,T}^f$ as the translation that uses $f$ to interpret atoms, $\tto_{F,T}$ to interpret $\tto$, and that commutes 
 with the propositional connectives.

We say that $\phi$ is \emph{$T$-valid} w.r.t. $F$ if, for all assignments $f$ of arithmetical sentences to
the propositional atoms, we have $T \vdash (\phi)_{F,T}^f$. We write $\Lambda_{F}(T)$ for the set of modal formulas that are $T$-valid w.r.t. $F$. Here are several examples of interpretations fitting in this framework, each of which yields a logic extending $\iam$. 

\begin{itemize}
\item  The most important one is provided by $\logptd{T}$, i.e.,  \emph{$\Delta$-preservativity}  for a theory $T$, where $\Delta$ is a class of 
 sentences. We assume that $\Delta$ is given by an elementary formula, say $\delta$, that arithmetically represents it. We define:
\begin{itemize}
\item
 $A \tto_{\Delta,T} B$ if, for all $\Delta$-sentences $S$, if
$T \vdash S \to A$, then $T \vdash S \to B$.

{\footnotesize We note that $\tto_{\Delta,T}$ is a rather uniform interpretation. There is an arithmetical formula $P(v_0,v_1,X,Y)$ with two free second order variables,
such that $B_0\tto_{\Delta,T}B_1 = P(\gnum{B_0},\gnum{B_1},\sigma_T,\delta)$.}
\end{itemize}
A minimal assumption on $\Delta$ is that it includes a sentence (equivalent to) $\top$. Let us call such $\Delta$ \emph{preservation-suitable}. 
  This ensures, in particular, that $\opr\phi$  defined as $\top \tto \phi$ still encodes provability. It is most common to fix $\Delta$ as $\Sigma^0_1$ \cite{viss:eval85,viss:prop94,viss:subs02,iemh:pres03,iemh:prop05}, in which case we  drop $\Delta$  (cf. \cite[\rfs{5}]{LitakV18:im} for the notation). 
  The relation of
$\Sigma^0_1$-preservativity is an important tool in the study of metatheory of {\sf HA}. In particular, constructive preservativity logic appears easier to analyze than its provability fragment. See  \cite{LitakV18:im} for further motivation.
\item  Classically, preservativity is seen as contraposed form of \emph{conservativity}, which in turn  is  classically equivalent to  \emph{interpretability}  \cite{bera:inte90,shav:rela88,japa:logi98,viss:over98,arte:prov04} \cite[C.3]{LitakV18:im}. It is also possible to investigate constructive interpretability logic, although it appears somewhat less well-motivated \cite[C.4]{LitakV18:im}.   
\item
Yet another example of a schematic interpretation is provided by the $\dotbox$-translation. That is, $\tto$ is interpreted as 
$\opr ((\phi \wedge \opr\phi) \to (\psi \wedge \opr\psi))$, where $\opr$ is the ordinary provability modality.
\item An example of an arithmetical interpretation not fitting well into the schematic framework sketched above is provided by the logic of 
admissible schemes: i.e., $\phi \tto \psi$ is valid if whenever an arithmetical instance of $\phi$ is provable, 
so is the corresponding instance of $\psi$. We postpone a detailed discussion to a later work. 
\end{itemize}

\section{Extension Stability} \label{sec:exsta}
The main focus of our paper is axiomatizing principles yielding explicit fixpoints in the logic  $\logt{T}$ of an arithmetical theory $T$ (relative to chosen $F$). Nevertheless, we have already mentioned that regardless of the interpretation of $\tto$, $\logt{\cdot}$ is not necessarily monotone: $T \subseteq T'$ does not entail that $\logt{T} \subseteq \logt{T'}$.  We thus would like to know that when our fixpoint principles hold in $\logt{T}$   for a base theory $T$, they also hold in $\logt{T'}$ whenever $T'$ is, say, a finite extension of $T$. 
 Otherwise, they would appear rather fickle. Thus, we are led to the notion of \emph{extension stability}. Apart from being central for our arithmetical (reverse) correspondence theory, it seems to be of interest for ordinary modal model theory; it turns out to be an overlooked generalization of the notion of \emph{subframe logic} \cite{Fine85:jsl,Wolter1993}. Recall that L\"ob-like logics, while being in general non-elementary, tend to be subframe, which above classical \lna{K4} ensures nice completeness properties. 
The details are discussed in a companion paper \cite{tlav19subf}; we just recapitulate basics here.

Given $\phi$ and  a fresh propositional variable $p$, define translation $\stex{p}{\phi}$ inductively as   
  commuting with the propositional variables 
  and the connectives of \ipc, with the $\tto$ clause being 
 \[\stex{p}{(\psi \tto \chi)} := ((p\to \stex{p}{\psi}) \tto (p\to \stex{p}{\chi})).\]
As $\opr\phi$ is $\top \tto \phi$, we get $\iam \vdash \stex{p}{(\opr\phi)} \iff \opr(p\to \stex{p}{\phi})$. 

A logic $\Lambda$ is \emph{extension stable} if, whenever $\Lambda \vdash \phi$ and $p$ not in $\phi$, we have
$\Lambda \vdash p \to \stex{p}{\phi}$. 

\begin{theorem}[\cite{tlav19subf}]
Suppose the logic $\Lambda$ is axiomatized over $\iam$ by $\Gamma$ \textup(w.r.t. the $\iam$-rules Modus Ponens and Necessitation\textup) 
and suppose that, for every $\phi \in \Gamma$ and $p$ not in $\phi$,
we have $\Lambda \vdash p \to \stex{p}{\phi}$. 
 Then  $\Lambda$ is extension stable.
\end{theorem}


\alnt{Perhaps also treat the case of closure under the substitution rule?}

\begin{definition}
We say that $F$ is an \emph{elegant interpretation} if for any recursively axiomatizable arithmetical theory $U$ and any arithmetical sentences $A,B,C$, we have that 
\[
U \vdash B \tto_{F,U+A} C \iff (A\to B) \tto_{F,U} (A \to C).
\]
\end{definition}

\begin{theorem}[\cite{tlav19subf}]
For any preservation-suitable $\Delta$, $\Delta$-preservativity is an elegant interpretation.
\end{theorem}

\begin{theorem}[\cite{tlav19subf}]\label{blauwsmurf}
Whenever $F$ is an elegant interpretation and $U$ is a recursively axiomatizable theory, $\mathfrak S(\Lambda_{F}(U))$ is of the form $\bigcap\limits_A\Lambda_{F}(U+A)$, where $A$ is ranging over all arithmetical sentences. Consequently, $\Lambda(U)$ is extension stable iff for all $A$ we have $\Lambda_F(U) \subseteq \Lambda_F(U + A)$.
\end{theorem}

\begin{remark} \tlnt{Do we keep this in this paper?}
We could also build the inverse: consider the $\stex{p}{\phi}$ for $\Lambda \vdash \phi$, where $p\not\in {\sf var}(\phi)$, and close off to make
it a logic. This gives an extension of $\Lambda$. 

Interestingly, this construction yields the well-known Kripke proof
of L\"ob's Theorem from G\"odel's theorem in the classical case: if we start with $\iam$ plus the excluded middle plus $\oco\top$,
we get $\loebb$. Similarly, we get $(\opr\phi \to \phi) \tto \phi$ from $\oco\top\tto \bot$. Regrettably, we need the excluded middle
here to get rid of negations that weaken the result constructively.
\end{remark}

\begin{exampleno} \tlnt{Do we keep this in this paper?}
The following examples are discussed in the companion paper \cite{tlav19subf}.
\begin{itemize}
\item[+] Positive examples of extension stability: the  interpretability logic {\sf ILM} of all classical sequential essentially reflexive theories  (\cite{shav:rela88}, \cite{bera:inte90}, \cite{japa:logi98}, or \cite{arte:prov04}) and
   the logic {\sf ILP} of all finitely axiomatized classical arithmetical theories  extending $\mathrm I\Delta_0+{\sf supexp}$  \cite{viss:inte90a}. 
   \item[-] Failure of extension stability: the interpretability logic of Elementary Arithmetic {\sf EA} \cite{kals:towa91}, the $\Sigma_1$-preservativity logic of \ha\ \cite{tlav19subf} \tlnt{do we claim this as a new result?}, closure under the Friedman translation 
 \cite{viss:prop94}.
 \end{itemize}
\end{exampleno}

\begin{question}
Axiomatize  the provability logic and the preservativity logic of all finite extensions of {\sf HA}. 
\end{question}


\section{{\iglam} Meets {\jv}}\label{iglamjv}
\tlnt{Again mention IDZ}

In this section we study a fixed point calculation that is due to de Jongh and Visser \cite{dejo:expl91},
 originally in the context of the classical interpretability logic {\sf IL}.
 It transfers naturally to the context of the Lewis arrow in the presence of axiom $\loeba$, as observed first by Iemhoff et al. \cite{iemh:prop05}. We will call the scheme that corresponds
to the de Jongh-Visser calculation \jv, show that it allows deriving \loeba\ (hence \loeba\ is precisely the logic of this construction) and that the resulting logic is extension stable.

We define {\iglam} as $\iam \thad \loeba$, where:
\begin{description}
\item[\loeba]
$(\opr\phi\to\phi)\tto \phi$.
\end{description}


 $\iam \thad \loebb \thad \quata$ coincides with \iglam \cite[Lemma 4.11]{LitakV18:im}. 

\begin{remark} \label{rem:presdi}
If we add classical logic to $\iglam$ and replace $\tto$ by the contraposed operator $\rhd$, we obtain the
interpretability logic {\sf IL}. We note that
 \di\ is not valid in contraposed {\sf IL}. See \cite{viss:inte90} and \cite{dejo:expl91}.
\end{remark}

We have the symmetric L\"ob rule:

\begin{description}
\item[\sloebra]
${\opr\phi \vdash \phi \iff \psi} \;\;\;\To \;\;\;  {\vdash \phi \strictiff \psi}$.
\end{description}

\begin{lemma}
The rule {\sloebra} is admissible for \iglam.
\end{lemma}

\begin{proof}
Suppose ${\opr\phi \vdash_{\iglam} \phi \iff \psi}$. It follows that $\psi \vdash_{\iglam} \opr\phi \to \phi$.
Hence, by {\na}, we find $\vdash_{\iglam} \psi \strictif (\opr\phi \to \phi)$ and, thus, by {\loeba},
$\vdash_{\iglam} \psi \strictif \phi$. 

For the converse direction, we have $\dotbox\phi \vdash_{\iglam} \psi$.
Hence, $\vdash_{\iglam} \dotbox\phi \tto \psi$.
Since, by {\quata}, we have $\vdash_{\iglam} \phi \tto \dotbox\phi$, by {\tr}, we
find $\vdash_{\iglam} \phi \tto \psi$.
\end{proof}


It has been noted  by Iemhoff et al. \cite{iemh:prop05} that the following result, first obtained in the classical setting \cite{dejo:expl91}, holds constructively as well.

\begin{theorem}\label{grotesmurf}
Let $\phi \deq \psi\tto\chi$ and $\theta := (\psi\opr\chi\top \tto \chi\top)$. Then, we have $\iglam \vdash \theta \iff \phi\theta$.
\end{theorem}

\begin{proof}
\ifconf\confbl\else
\ExecuteMetaData[apptwocalc.tex]{grotesmurf}
\fi
\end{proof}

We note that $\theta$ can also be written as $(\psi\opr p \tto p)\chi \top$ or as $(\psi\opr \chi \tto \chi)\top$. 

We define the principle {\jv} as follows:
\begin{description}
\item[\jv]
$(\psi\opr\chi\top \tto \chi\top) \;\; \iff 
(\psi \tto \chi)(\psi\opr\chi\top \tto \chi\top)$
\end{description}

We have seen that $\iglam \vdash \jv$. We now show that, conversely, $\iam\thad\jv \vdash \loeba$,
and, hence, $\iam\thad\jv$ coincides with 
{\iglam} .

\begin{theorem}\label{olijkesmurf}
$\iam\thad\jv \vdash \loeba$.
\end{theorem}

\begin{proof}
The {\jv}-fixed point of 
$\opr(p \to \phi)$, where $p$ does not occur in $\phi$, is  $\opr(\top \to \phi)$.
This tells us that we, $\iam\thad\jv$-verifiably, have \loebb. 
It follows, by Lemma~\ref{brilsmurf}, that we  have uniqueness of fixed points in $\iam\thad\jv$. 
Now consider the formula $(p \to \phi) \tto \phi$, where $p$ does not occur in $\phi$. On the one hand,
$\top$ is a fixed point of this formula even over \iam. On the other hand, {\jv} gives us
$(\opr \phi \to \phi) \tto \phi$ as a fixed point. By uniqueness, we find that  $(\opr \phi \to \phi) \tto \phi$ is $\iam\thad\jv$-provable.
\end{proof}

We note that, in the proof of Theorem~\ref{olijkesmurf}, we used that we are allowed to choose $p$ locally. 
We summarize Theorems~\ref{grotesmurf} and \ref{olijkesmurf} in one statement.

\begin{corollary}\label{samenopwegsmurf}
$\iglam$ coincides with $\iam\thad\jv$. 
\end{corollary}

We see that it follows from  that $\iam\thad\jv$ is indeed a logic in that it is closed under substitution. This could have been proven, of course,
without the detour over the characterization.

\begin{theorem} \label{th:jvstable}
$\iam \thad \jv$ is extension stable.
\end{theorem}

\begin{proof}
It is sufficient to show that {\iglam} is extension stable. We leave the verification that {\iam} is extension stable to the reader.
We have: 
\[
 \stex{p}{((\opr \phi \to \phi) \tto \phi)} \quad = \quad  (p \to (\opr (p \to \stex{p}{\phi}) \to \stex{p}{\phi})) \tto (p\to \stex{p}{\phi}).
 \]
Clearly, 
\[
 (\dag) \qquad \iam\vdash \quad \stex{p}{((\opr \phi \to \phi) \tto \phi)} \quad \iff \quad (\opr (p \to \stex{p}{\phi}) \to (p \to \stex{p}{\phi})) \tto (p\to \stex{p}{\phi}).
 \]
The right-hand-side of (\dag) is an instance of \loeba.
\end{proof}

We end this section by comparing {\iglam} to the stronger theory \iglwm.
We remind the reader that the principle {\weak} can be equivalently written as 

\begin{description}
\item[\weak]
$(\phi \tto \psi) \to ((\opr \psi \to \phi)\tto \psi)$.
\end{description}

 We define $\iglwm := \iam\thad\weak$. It is known that this is a stronger logic than $\iam$ \cite{iemh:prop05}. To separate {\iglam} from \iglwm, we use here a lemma that is an adaptation
 of an example from de Jongh and Visser  \cite[p. 48]{dejo:expl91}.
 
 \begin{lemma}\label{handigesmurf}
 $\iglam \nvdash (\opr\bot\tto\bot) \to \opr\bot$.
 \end{lemma}
 
 \begin{proof}
  We can show that $\iglam \nvdash (\opr\bot\tto\bot) \to \opr\bot$ on a three point Kripke model. See \cite[Example 6.7]{LitakV18:im}.
 \end{proof}
 
 \begin{theorem}\label{babysmurf}
 {\iglwm} strictly extends {\iglam}.
 \end{theorem}
 
 \begin{proof}
 To see that  {\iglwm} extends {\iglam}, we note that specializing {\weak} by taking $\psi:= \phi$
 gives us {\loeba} over \iam.
 
We take $\phi := \opr\bot$ and $\psi := \bot$ in \iglwm. This gives us $\iglwm \vdash (\opr\bot\tto\bot) \to \opr\bot$.
By Lemma~\ref{handigesmurf}, $\iglam \nvdash (\opr\bot\tto\bot) \to \opr\bot$.
 \end{proof}


\section{\js} \label{sec:js}

We proceed with the study of explicit fixed points in the style of de Jongh and Sambin.
We define {\js} as follows.
\begin{description}
\item[\js]
$(\psi \tto \chi)\top \iff (\psi \tto \chi)(\psi \tto \chi)\top$.\\
Written differently:\\
$(\psi\top \tto \chi\top) \iff (\psi \tto \chi)(\psi\top \tto \chi\top)$
\end{description}

The principle {\js} is \emph{prima facie} simpler than {\jv}, however, the impression is somewhat misleading.
As we will see, over $\iam$, neither principle implies the other.

We begin with a natural preliminary result:

\begin{theorem}
$\iam \thad \js$ is a logic.
\end{theorem}

We remind the reader that the results of applying $\thad$ are not automatically closed under substitution.

\begin{proof}
Suppose that $q$ is distinct from the variables in $\psi$ and $\chi$ and that $q$ is distinct from the substitution variable $p$.
Let $\nu$ be given. By possibly renaming the substitution variable $p$ we can arrange that $p$ does not occur in $\nu$.
We have:
\begin{multline*}
\app{q}{((\psi \tto \chi)\top \iff (\psi \tto \chi)(\psi \tto \chi)\top)}{\nu}  = \\ 
(\app{q}{\psi}{\nu} \tto \app{q}{\chi}{\nu})\top \iff \\ ((\app{q}{\psi}{\nu} \tto \app{q}{\chi}{\nu})(\app{q}{\psi}{\nu} \tto \app{q}{\chi}{\nu})\top).
\end{multline*}
The result of substitution is again an instance of \js.
\end{proof}

\begin{theorem} \label{th:jsstable}
$\iam\thad\js$ is extension stable.
\end{theorem}

\begin{proof}
Let $q$ be distinct from the variables in $\psi$ and $\chi$ and from the substitution variable $p$.

We have: 
{\footnotesize
\begin{multline*}
\stex{q}{((\psi \tto \chi)\top \iff (\psi \tto \chi)(\psi \tto \chi)\top)}  = \\ 
(((q \to \stex{q}{\psi}) \tto (q \to\stex{q}{\chi}))\top \iff ((q\to \stex{q}{\psi}) \tto \\
(q\to \stex{q}{\chi}))((q\to\stex{q}{\psi}) \tto (q\to \stex{q}{\chi}))\top).
\end{multline*}
}
The resulting formula is again an instance of \js.
\end{proof}

The proof of the following theorem is essentially due to de Jongh and Visser \cite{dejo:expl91}. In the intuitionistic setting, it has been also noted by Iemhoff et al. \cite{iemh:prop05}.

\begin{theorem}\label{knutselsmurf}
$\iglwm \vdash \js$.
\end{theorem}

\begin{proof}
Since {\iglwm} extends {\iglam}, we have {\jv} in  {\iglwm}. Moreover:
\begin{eqnarray*}
\iglwm \vdash (\psi \opr \chi\top \tto \chi \top) & \iff &
((\opr \chi \top \to \psi \opr \chi\top) \tto \chi \top) \\
& \iff & ((\opr \chi \top \to \psi \top) \tto \chi \top) \\
& \iff & ( \psi \top \tto \chi \top) 
\end{eqnarray*}
Here the first and third equivalences are instances of \weak\ and the second equivalence is an instance
of the substitution principle {\sf Su1}, noting that $\opr\chi\top$ is equivalent to $\dotbox(\opr\chi\top \iff \top)$.
\end{proof}

Now consider the principle 
\begin{description}
\item[\pers]
$(\phi\tto\psi) \to \opr(\phi\tto\psi)$.
\end{description}
We define $\iglpm:=\iam\thad\loebb\thad\pers$.

We prove {\js} in {\iglpm}.
The development here is essentially the one from Smory\'nski's \cite[Chapter 4]{smor:self85}.

\begin{theorem}\label{smurfin}
$\iglpm \vdash \js$.
\end{theorem}


\begin{proof}
We reason in $\iglpm$. 

Suppose $(\psi\tto \chi)\top$. Then, $\dotbox((\psi\tto \chi)\top \iff \top)$. Hence, 
$(\psi\tto \chi)(\psi \tto \chi)\top$.

Conversely, suppose $(\psi\tto \chi)(\psi \tto \chi)\top$ and $\opr(\psi\tto\chi)\top$.
From the latter we can derive 
\[
\opr((\psi\tto \chi)\top \iff \top).
\] 
Since in $(\psi\tto \chi)(\psi \tto \chi)\top$ the formula $(\psi \tto \chi)\top$ occurs only
modalized, we find: $(\psi\tto \chi)\top$. By L\"ob's Rule, we may drop the assumption
 $\opr(\psi\tto\chi)\top$.
\end{proof}

We end this section by showing some incomparability results.

\begin{lemma}\label{slimmesmurf}
$\iglp \nvdash \quata$. Consequently, \iglp\ is not a subtheory of \iglam.
\end{lemma}

\begin{proof}
If over
{\iglb} we define $\phi\tto \psi$ as $\opr(\phi\to\psi)$ this validates an extension of \iglp. However,
{\quata} would translate to $\opr(\phi\to \opr\phi)$, which is obviously not in \iglp, or even in the classical version theoreof; 
%
  let us recall a trivial proof. We can consider the 
Kripke model  on three different points $a,b,c$. We take ${\sqsubset} := \verz{\tupel{a,b},\tupel{a,c},\tupel{b,c}}$
and ${\preceq} = \verz{\tupel{a,a},\tupel{b,b},\tupel{c,c}}$ (that is, a classical model). Let $b \Vdash p$ and $c \nVdash p$. Then, the model
satisfies {\iglp} but $a\nVdash p\tto \opr p$.
\end{proof}

In fact, the non-containment goes both ways, i.e., extends to incomparability. 

\begin{theorem}\label{querusmurf}
 {\iglpm} is  not contained in {\iglw} \bro and a fortiori not in \igla\brc.
\end{theorem}

\begin{proof}
One possible argument is via interpretability logic, however this only works for non-containment in \iglwm.\footnote{Suppose we extend {\iglwm} and {\iglpm} with classical logic, and we add {\quata}, and  we switch to the contraposed reading $\rhd$.
We thus obtain, respectively, the interpretability logics {\sf ILW} and {\sf ILP}. It is well known that  {\sf ILP} strictly extends {\sf ILW}.
So it follows that  {\iglpm} is not a subtheory of {\iglwm}. Hence, a fortiori, it is not a subtheory of \iglam.}
Thus, consider the following Kripke frame: 
\[
a \sqsubset b \prec c \prec d, \quad 
a  \sqsubset c, \qquad 
b \sqsubset d,  \quad
b \prec d. 
\]
We set that $d\Vdash p$ and $p$ does not hold anywhere else.
This model is easily seen to be supergathering but $a$ does not validate 
 $p \tto \bot \to \opr(p \tto \bot)$.
\end{proof}

%
%

We note that ${\iglwm} \subseteq {\iglpm}\thad \quata  = \iglam \thad \weak$.

Since $\iam\thad\js$ is contained in the incomparable theories {\iglwm} and {\iglpm}, it is strictly contained
is both.

\begin{theorem}\label{grolsmurf}
The theory $\iam\thad\js$ is strictly contained in {\iglwm} and {\iglpm}.
\end{theorem}

\begin{theorem}\label{medicijnsmurf}
{\iglam} is incomparable to $\iam\thad\js$.
\end{theorem}

\begin{proof}
{\iglam} proves {\quata}, but, by Lemma~\ref{slimmesmurf}, {\iglpm} does not. So, a fortiori,
$\iam\thad\js$ does not prove \quata.

By considering the fixed point of $p\tto \bot$, we see that $\iam\thad\js$ proves $(\opr\bot\tto\bot) \to \opr\bot$.
However, by Lemma~\ref{handigesmurf}, {\iglam} does not.
\end{proof}

\begin{question}
Can we give an interesting characterization of $\iam\thad\js$? E.g. can we derive it
from finitely many schematic fixed point equations? Or can we refute that this is possible? In the rest of this paper we will make the question a bit more specific.
\end{question}

\section{The Join of {\jv} and {\js}} \label{sec:join}
What happens if we add an axiom that says that the \jv-fixed point of $\psi\tto\chi$ is equal to the \js-fixed point?

Consider the scheme:
\begin{description}
\item[{\sf X}]
$(\psi\opr\chi\top \tto \chi\top) \iff (\psi\tto\chi)\top$. 
\end{description}

\begin{theorem}
$\iam\thad {\sf X} \vdash \weak$.
\end{theorem}

\begin{proof}
Consider the formula $(p\to \phi) \tto\psi$, where $p$ does not occur in $\phi$ and $\psi$.
The {\jv}-fixed point of $(p\to \phi) \tto\psi$ is $(\opr\psi \to \phi) \tto \psi$ and the
{\js}-fixed point is modulo \iam-provable equivalence $\phi\tto \psi$.
\end{proof}

\begin{corollary}
The following theories are equal: 
$\iam \thad \jv \thad \js$, $\iam\thad{\sf X}$ and {\iglwm}.
\end{corollary}

\begin{proof}
Since $\iam \thad \jv \thad \js$ extends $\iglbm$, we have uniqueness of fixed points and hence {\sf X}.
Moreover, $\iam \thad {\sf X} \vdash \weak$. Finally, we have already seen that $\iglwm$ proves
{\js} and {\jv}.
\end{proof}

\begin{corollary}
 {\iglpm} does not prove {\jv} and {\iglam} does not prove {\js}.
 \end{corollary}

\begin{question}
Can we find a logic below {\iglam} and {\iglpm}  that extends {\iglbm} 
with explicit fixed points? We note that these fixed points would collapse into the
\jv-fixed points on extension to {\iglam} and into the \js-fixed points upon extension to
{\iglpm} or {\iglwm}.
\end{question}


\section{Four Salient Fixpoints: Subtheories of $\iam\thad\js$}
\label{sec:jscloser}

There are four approaches to finding a nice axiomatization of $\iam\thad\js$:
(i) a brilliant flash of insight, (ii) trying out some salient fixed points and see whether they
already give us the desired scheme, (iii) analyzing what is shared between the \iglwm-proof of {\js} and
the \iglpm-proof and (iv) investigating what {\js} means in Kripke models.

In this section, we contribute to approach (iii), by computing four salient \js-fixed points.



\newcommand{\brlt}{\newline\quad} 

\begin{figure}
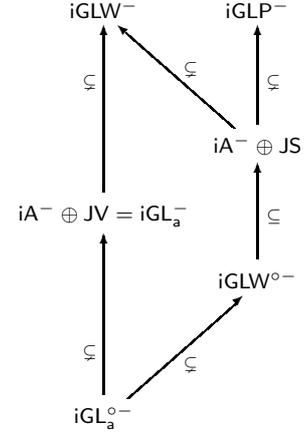

\hrule \sepo
\rowcolors{1}{}{}
\footnotesize
\begin{tabular}{lL{3.5cm}L{3.9cm}C{3.3cm}} 
$\jv$ & $(\psi\opr\chi\top \tto \chi\top)$ \newline $\iff 
(\psi \tto \chi)(\psi\opr\chi\top \tto \chi\top)$ & the de Jongh-Visser scheme  of explicit fixpoints  &
\multirow{5}{*}{
\begin{diagram}[nohug]
\iglwm  && \iglpm \\
\uTo^{\subsetneq} &  \luTo(2,2)^\subsetneq & \uTo_{\subsetneq} \\
  && \iam \thad \js \\
 \iam\thad \jv = \iglam && \uTo_\subseteq \\ 
 && \iglwmo \\
  \uTo^\subsetneq  & \ruTo_\subsetneq & \\
  \iglamo && 
\end{diagram}
} 
\\ 
\rowcolor{light-gray}\sepo
$\loeba$ & $(\opr\phi\to\phi)\tto \phi$ & its direct axiomatization 
\\ \cline{1-3}
$\js$ & $(\psi\top \tto \chi\top)$ \brlt$\iff (\psi \tto \chi)(\psi\top \tto \chi\top)$ & the de Jongh-Sambin scheme  of explicit fixpoints \\
\rowcolor{light-gray}
$\pers$ & $(\phi\tto\psi) \to \opr(\phi\tto\psi)$ & one of principles ensuring it holds over $\loebb$ 
\\
$\weako$ &  $((\phi \wedge (\phi \tto \psi))  \tto \psi)$ \brlt$\to (\phi\tto\psi)$ & the most powerful consequence of \js\ we know \\ 
\rowcolor{light-gray}\sepo
$\weakst$ & $(\phi \tto \psi)$ \brlt$\to (((\phi \tto \psi) \to \phi) \tto \psi)$ & an alternative axiomatization for $\weako$ over $\loebb$ \\
\cline{1-3}
{\sf X} & $(\psi\opr\chi\top \tto \chi\top)$ \brlt$\iff (\psi\tto\chi)\top$ & the scheme \newline collapsing \jv\ and \js\  \\ \rowcolor{light-gray}\sepo
\weak & $(\phi \tto \psi)  $ \brlt$\to ((\opr \psi \to \phi)\tto \psi)$ & its direct axiomatization \\ \cline{1-3}
$\loebao$ & $(\phi \tto ((\phi \tto \psi ) \to \psi))$ \brlt $\tto (\phi\tto \psi)$ & a principle holding both under \js\ and \jv\ \\ 
\rowcolor{light-gray}
\sepo
$\quatao$ & 
 $(\phi\tto \psi) \to (\phi \tto (\phi \tto \psi ))$ &  a principle deriving $\loebao$ over $\loebb$ \\[\tbskip]  
\end{tabular}

\hspace{0.4cm}

\caption{Relationships between  logics axiomatized by either fixpoint principles or axioms obtained as fixpoints of simple formulas \label{fig:logics}} \vspace{0.2cm}
\hrule
\end{figure}

\subsection{The Principle \weakst}

The \js-fixed point of $(p\to \phi) \tto \psi$ is modulo $\iam$-provable equivalence $\phi\tto \psi$. The fixed point equation is 
$(\phi \tto \psi) \iff (((\phi \tto \psi) \to \phi) \tto \psi)$.
Modulo {\iam}-provable equivalence this equation simplifies to the principle {\weakst}:
\begin{description}
\item[\weakst]
$(\phi \tto \psi) \to (((\phi \tto \psi) \to \phi) \tto \psi)$
\end{description}

Here is a first insight about \weakst.

 \begin{theorem}
 Over $\ia\thad\cpc$, the principle $\weakst$ axiomatizes the same logic as  $\quatb$.
 \end{theorem}
 
 \begin{proof}
 Over $\ia\thad \cpc$, $\phi\tto\psi$ collapses to $\opr(\phi\to \psi)$.
 See \cite[Lemma 4.6]{LitakV18:im}. This tells us that over $\ia\thad \cpc$ the principle
 \weakst\ reduces to:  \[(\dag) \;\;\; \opr (\phi \to \psi) \to \opr ((\opr(\phi \to \psi) \to \phi) \to \psi).\]
Putting $\phi := \neg\,\chi$ and $\psi :=\bot$, (\dag) gives us classically \quatb, i.e., 
$\opr\chi\to \opr\opr\chi$. Conversely, we easily get (\dag) from \quatb.
 \end{proof}
 
 \begin{corollary}
$\iam\thad\weakst$ does not imply \loebb.
\end{corollary}

\subsection{The Principle {\weako}}

The \js-fixed point of $(p \wedge \phi) \tto \psi$ is modulo \iam-provable equivalence $\phi\tto\psi$. The fixed point equation is
$ (\phi \tto \psi) \iff  ((\phi \wedge (\phi \tto \psi))  \tto \psi)$.   
Modulo {\iam}-provable equivalence this equation simplifies to the principle {\weako}:
\begin{description}
\item[\weako]
$((\phi \wedge (\phi \tto \psi))  \tto \psi) \to (\phi\tto\psi)$
\end{description}

We note that, if we put $\phi := \top$ in \weako, we get $(\opr \psi \tto \psi) \to \opr\psi$.  This last, by \iam-reasoning, implies
{\loebb}. 
  
\begin{theorem}
 Over $\ia\thad\cpc$, the principle $\weako$ axiomatizes the same logic as  $\loebb$.
 \end{theorem}

\begin{proof}
By  \cite[Lemma 4.6]{LitakV18:im}, \weako\ reduces to:
\[ (\dag)\;\;\; \opr((\phi \wedge \opr(\phi \to \psi))  \to \psi) \to \opr (\phi\to\psi) \] 
By putting $\phi := \top$ we find \loebb\ modulo \ia-provable equivalence.
Conversely, it is easily seen that (\dag) is \ia-provably equivalent to an instance of \loebb.
\end{proof}


\begin{theorem} \label{wcircderiv}
$\iam \thad \loebb \thad \weakst\vdash \weako$. 
\end{theorem}

\begin{proof}
\ifconf\confbl\else
\ExecuteMetaData[apptwocalc.tex]{wcircderiv}
\fi
\end{proof}
 
 \begin{theorem} \label{weakstderiv}
$\iam \thad \weako\vdash \weakst$. 
\end{theorem}

\begin{proof}
\ifconf\confbl\else
\ExecuteMetaData[apptwocalc.tex]{weakstderiv}
\fi
\end{proof}

 We define $\iglwmo := \iam\thad \weako$. 
 We have already shown:
 
 \begin{theorem}
 \begin{enumerate}[i.]
 \item
  {\iglwmo} strictly extends $\iam\thad\weakst$.
  \item
  {\iglwmo} is equal to $\iam\thad\weakst\thad \loebb$.
  \item
 {\iglwmo} is contained in $\iam\thad\js$.
 \end{enumerate}
 \end{theorem}
 
 We also have:
 \begin{theorem}
  {\iglwmo} is incomparable to $\iglam$.
 \end{theorem}
 
 \begin{proof}
 Since {\iglwmo} is contained in {\iglpm}, it cannot contain \iglam.
 Conversely, putting $\phi := \top$ and $\psi:=\bot$, we see that
 ${\iglwmo} \vdash (\opr\bot\tto\bot) \to \opr\bot$. But by Lemma~\ref{handigesmurf}, the logic {\iglam} does not prove this.
 \end{proof}
  
  \subsection{The Principle {\loebao}}
  
 The \js-fixed point of $\phi\tto(p\to \psi)$, where $p$ does not occur in $\phi$ and $\psi$ is, modulo \iam-provable equivalence,
 $\phi\tto \psi$.
The fixed point equation is:
$ (\phi\tto \psi) \iff (\phi \tto ((\phi \tto \psi ) \to \psi))$.
Modulo \iam-provable equivalence this reduces to:
\begin{description}
\item[\loebao]
$  (\phi \tto ((\phi \tto \psi ) \to \psi)) \tto (\phi\tto \psi)$.
\end{description}
 
 We note that putting $\phi := \top$ in {\loebao} gives us precisely \loebb.
 
 We define $\iglamo := \iam\thad\loebao$.
 
 \begin{theorem}\label{smurferella}
 $\iglwmo \vdash \loebao$. In other words, {\iglwmo} extends \iglamo.
 \end{theorem}
 
 \begin{proof}
 In {\iglwmo} we have: \qedright
 \begin{eqnarray*}
 (\phi \tto ((\phi \tto \psi ) \to \psi)) 
 & \to & ((\phi \wedge (\phi \tto \psi)) \tto \psi) \\
 & \to & \phi \tto \psi
 \end{eqnarray*}
 \end{proof}
 
 \subsection{The Principle {\quatao}}

The \js-fixed point of $\phi\tto(p\wedge \psi)$, where $p$ does not occur in $\phi$ and $\psi$ is, modulo \iam-provable equivalence, $\phi\tto\psi$.
 The fixed point equation is \[ (\phi\tto \psi) \iff (\phi \tto ((\phi \tto \psi ) \wedge \psi)).\]
  Over $\iam$ this simplifies to the following principle:
 \begin{description}
 \item[\quatao]
 $ (\phi\tto \psi) \to (\phi \tto (\phi \tto \psi ))$
 \end{description}
Putting $\phi := \top$ in {\quatao} gives us precisely \quatb.

\begin{theorem}
$\iglamo \vdash \quatao$.
\end{theorem}

\begin{proof}
We reason in \iglamo. \qedright
\begin{eqnarray*}
\phi\tto \psi & \to & \phi \tto ((\phi \tto (\psi \wedge (\phi \tto \psi))) \to (\psi \wedge (\phi \tto \psi))) \\
& \to & \phi \tto (\psi \wedge (\phi \tto \psi)) \\
& \tto & \phi \tto (\phi \tto \psi)
\end{eqnarray*}
\end{proof}

\begin{theorem}
$\iam\thad\loebb \thad\quatao \vdash \loebao$ and hence $\iam\thad\loebb \thad\quatao$ coincides with 
 \iglamo.
\end{theorem}

\begin{proof}
Let $\gamma := ((\phi\tto\psi)\to \psi)$.
We have over $\iam \thad \loebb \thad\quatao$:

{\footnotesize
$(\phi \tto \gamma) \wedge \opr ((\phi\tto \gamma) \to (\phi\tto\psi)) \quad \to$  
\begin{eqnarray*}
& \to & (\phi \tto \gamma) \wedge \opr ((\phi\tto \gamma) \to (\phi\tto\psi)) \wedge (\phi \tto(\phi \tto \gamma)) \\
& \to & (\phi\tto \gamma) \wedge (\phi \tto (\phi\tto \psi)) \\
& \to &(\phi\tto \psi) 
\end{eqnarray*}
}
We find the desired result by applying L\"ob's Rule.
\end{proof}

We note that we can interpret the unimodal L\"ob system  in \iglamo\ by reading $\phi\tto(\cdot)$ for the box.

\begin{theorem}\label{smurfette} 
$\iglam \vdash \loebao$, i.e., {\iglam} extends \iglamo.
\end{theorem}

\begin{proof}
\qedright
We reason in \iglam.
\begin{eqnarray*}
(\phi \tto ((\phi\tto \psi) \to \psi)) & \to & (\phi \tto (\opr \psi \to \psi)) \\
& \to & (\phi \tto\psi)
\end{eqnarray*}
\end{proof}

\begin{corollary}
{\iglamo} is strictly contained in {\iglwmo} and in \iglam.
\end{corollary}

\begin{proof}
The strictness follows from the fact that {\iglamo} is contained in the incomparable theories {\iglwmo} and \iglam.
\end{proof}


\begin{question}
Is $\iglwmo$ equal to $\iam\thad\js$ and do systems such as {\iglwmo} or perhaps even $\iglamo$ have explicit fixed points at all?  
\end{question}


Figure \ref{fig:logics} summarizes the relationships between the theories studied.
Theories that are incomparable in the picture are really incomparable. As we see the main
question left open in the picture is the yes-or-no identity of $\iam\thad\js$ and \iglwmo.

\section{Correspondence and Non-conservativity} \label{sec:corresp}


We can now extend the list of known correspondences given in Theorem \ref{th:minco} and  Figure \ref{tab:compl} to several additional principles, whose importance has been highlighted by our study.
%

\begin{theorem}
In Kripke semantics, {\pers} corresponds to the transitivity of $\sqsubset$.
\end{theorem}

\begin{proof}
Clearly transitivity implies \pers. Conversely, consider a frame $\mathcal F$ with
nodes $a\sqsubset b \sqsubset c$ and $a \not \sqsubset c$. We define
$x \Vdash p$ iff $c\preceq x$, and $x\Vdash q$ iff $x \not\preceq c$.
Clearly, $c \Vdash p$ and $c\nVdash q$, so $a \nVdash \opr(p\tto q)$.
On the other hand, consider any $z$ with $a\sqsubset z\Vdash p$.
It follows that $c\preceq z$. By our assumption $z \neq c$, so $z \not \preceq c$.
Ergo $z \Vdash q$. Thus, $a \Vdash p\tto q$.
\end{proof}

\begin{theorem}
In Kripke semantics, {\quatao} corresponds to gather-transitivity: if $x\sqsubset y \sqsubset z$, then $x\sqsubset z$ or $y \preceq z$.
\end{theorem}

\begin{proof}
Suppose $\mathcal F$ is gather-transitive. Consider any model on $\mathcal F$. Suppose $a \Vdash \phi \tto \psi$.
Suppose we have $a \sqsubset b \sqsubset c$ and $b \Vdash \phi$ and $c\Vdash \phi$. In case $a\sqsubset c$, we have
$c\Vdash\psi$ and in case $b\preceq c$ we also have $c\Vdash \psi$. Done.

Conversely, consider a frame $\mathcal F$ with $a \sqsubset b\sqsubset c$ and (i) $a\not\sqsubset c$ and (ii) $b\not\preceq c$.
We define
$x \Vdash p$ iff $b \preceq x$ or $c\preceq x$, and $x\Vdash q$ iff $x \not\preceq c$.
Clearly, $b\vdash p$ and $c \Vdash p$ and $c\nVdash q$, so $a \nVdash \opr(p\tto q)$.
On the other hand, consider any $z$ with $a\sqsubset z\Vdash p$.
It follows that $b \preceq z$ or $c\preceq z$. In case we would have $z\preceq c$, it follows that either $b \preceq c$ or $z=c$.
The first possibility is excluded by (ii). The second possibility is excluded since it would follow that $a\sqsubset c$ contradicting (i).
 Ergo, $z\not \preceq c$ and, so, $z \Vdash q$. Thus, $a \Vdash p\tto q$.
\end{proof}

In the presence of \di, we also have 

\begin{description}
\item[\qquatao] $(\phi\tto(\psi\tto\chi)) \to (\phi\tto(\psi\tto(\phi\tto (\psi\tto \chi))))$. 
\end{description}
\begin{theorem} \label{ququatao}
 $\ia\thad\quatao \vdash \qquatao$. 
\end{theorem}

\begin{proof}
\ifconf\confbl\else
\ExecuteMetaData[apptwocalc.tex]{ququatao}
\fi
\end{proof}

Nevertheless, the above derivation crucially uses \di. It is a cautionary tale that the Kripke correspondence conditions lose some of their importance for extensions of $\iam$ which do not extend $\ia$.

\begin{theorem} \label{iamququatao}
 $\iam\thad\quatao \nvdash \qquatao$. 
\end{theorem}

\begin{proof}
We use algebraic semantics for this purpose. A smallest possible, 6-element countermodel found by the program Mace4 \cite{McCune03} is provided in the appendix. Its Heyting reduct looks as follows:
\[   
\vcenter{
    \xymatrix@-1pc{
    & & a_1 = \top & &\\
    \phi = & a_2 \ar@{-}[ur] & & a_4  \ar@{-}[ul] & = \psi \\
    & & a_5 \ar@{-}[ul] \ar@{-}[ur] & a_3 \ar@{-}[u]  & = \chi \\
    & & a_0 = \bot \ar@{-}[u] \ar@{-}[ur] &
      } }
 \]     
 Some key fact about the definition of $\tto$ are:
 \begin{itemize}
 \item For every $a \neq \top$, $\top \tto a = \bot$; 
 \item $a_3 \tto a_2 = \bot$;
 \item $a_5 \tto \bot = a_5 \tto a_3  = a_4 \tto a_3  = a_4$;
 \item $a_2 \tto \bot = a_2 \tto a_3 = a_5$;
 \item $a_2 \tto a_4 = a_2 \tto a_5 = a_2$.
 \end{itemize} 
 Note that via $a_4 \tto a_3$ and $a_2 \tto a_3$ one can disprove $\di$.
\end{proof}







 
 
 \section{Explicit Fixpoints and the Beth Property} \label{sec:beth}
 
 It has been shown first by Maximova \cite{Maksimova1989} that explicit fixpoints yield the Beth definability property and subsequent literature provides further investigation of this phenomenon \cite{ArecesHJ98,hoog:defi01}. For constructive $\tto$-logics, this argument has already been made by Iemhoff et al. \cite[\S\ 4.3]{iemh:prop05}. Hence, we simply state as an observation that any logic which extends either $\iam\thad\js$ or $\loglg = \iam\thad\jv$ enjoys this property.
 
 \section{Conclusions and Future Work} \label{sec:conclusions}
 
 This paper has studied the surprisingly complex first stage of construction of guarded fixpoints: for formulas whose principal connective is $\tto$ itself. Even its restricted scope has left some questions unanswered.   
 In follow-up work, 
we will investigate the question of how to extend
the definability result from a restricted class of formulas to a wider one, including the computational cost of such extensions; syntactic setup necessary for that purpose (inspired by and generalizing that of  \cite{dejo:expl91}; the issue of commuting resulting fixpoints operators with substitutions has intriguing Beck-Chevalley aspects) would distract too much from the main goal of the present paper. 
In particular, we have mentioned that classically, it is possible to use  elimination of guarded fixpoints for elimination of positive  fixpoints of ordinary $\mu$-calculus  \cite{bent:moda06,viss:lobs05}. Furthermore,  is possible to define a portmanteau notion of \emph{semipositive} formulas (where every occurrence of $p$ is either guarded or positive), which in general allow explicit \emph{locally minimal} fixpoints \cite{viss:lobs05}. It is possible to generalize such results to the intuitionistic setting, but the techniques and results involved are more demanding. 




\bibliographystyle{alpha}
\bibliography{provintconf,loebmuintlewis,intmodlewis}

\appendix
\ifconf

\input{apptwocalc}
\else\fi
\input{fourfour}

\end{document}

%% file: apptwocalc.tex
\section*{Proof of Lemma \ref{sub1sub2}}

\begin{proof}
We prove {\sf Sub}1 by induction on $\chi$. The only interesting case is where $\chi = \nu_0 \tto \nu_1$. By the Induction Hypothesis,  for $i=0,1$, we have:
\[\iam\thad \quatb \vdash \dotbox(\phi \iff \psi) \to (\nu_i\phi \iff \nu_i\psi).\]
It follows that: \[\iam\thad \quatb \vdash \dotbox(\phi \iff \psi) \to \opr(\nu_i\phi \iff \nu_i\psi)\]
 and, hence that:
\[\iam\thad \quatb \vdash \dotbox(\phi \iff \psi) \to (\nu_i\phi \ifff \nu_i\psi).\]
The desired result is now immediate. \neskip

We prove {\sf Sub2}. Suppose $\chi = \nu_0 \tto \nu_1$. We have, by {\sf Sub1}, 
\[ \iam\thad \quatb \vdash \dotbox(\phi \iff \psi) \to (\nu_i\phi \iff \nu_i\psi).\]
Hence, 
\[ \iam\thad \quatb \vdash \opr (\dotbox(\phi \iff \psi) \to (\nu_i\phi \iff \nu_i\psi)).\]
So, 
\[ \iam\thad \quatb \vdash \opr\dotbox(\phi \iff \psi) \to \opr(\nu_i\phi \iff \nu_i\psi).\]
Ergo, since $\iam\thad \quatb \vdash \opr\phi \to \opr\dotbox \phi$, we find:
\[ \iam\thad \quatb \vdash \opr(\phi \iff \psi) \to (\nu_i\phi \ifff \nu_i\psi).\]
Thus the desired result follows for the case that $\chi = \nu_0 \tto \nu_1$.
The remaining cases are by induction on outer non-modal propositional connectives.
\end{proof}


\section*{Proof of Lemma \ref{brilsmurf}}
\begin{proof}
We use the familiar fact that we have $\quatb$ in $\iam\thad \loebb$. It is easy to see that we have the strengthened L\"ob's Rule in $\iam\thad \loebb$, to wit:
\begin{center}
if $ \iam\thad \loebb \vdash (\bigwedge_{i<n}\dotbox \phi_i \wedge \opr \psi) \to \psi$, then $ \iam\thad \loebb \vdash \bigwedge_{i<n}\dotbox \phi_i \to \psi$. 
\end{center}
We use this rule in our proof. We only prove (a), as  items (b) and (c) follow then immediately. We reason in $ \iam\thad \loebb$. Suppose $\dotbox (r\iff \chi r)$ and $\dotbox (q\iff \chi q)$ and $\opr(r\iff q)$.  Since we have $\quatb$, we also have {\sf Sub}2 and, hence, $\chi r \iff \chi q$. Thus we find $r\iff q$. By the strengthened L\"ob's Rule, we may conclude $r\iff q$ without the assumption of  $\opr(r\iff q)$.
\end{proof}


\section*{Proof of Theorem \ref{algsound}}

\begin{proof}
For the first part, i.e., soundness: closure under substitution is trivial, Modus Ponens and \ipc\ axioms are standard using the fact that our algebras have Heyting reducts. We get \tr\ via the validity of \lna{CT} and \ka\ via the validity of \lna{CK}. For $\na$, assume $\gA, v \Vdash \phi \to \psi$. By standard facts regarding Heyting algebras, $\hat{v}(\phi) \leq \hat{v}(\psi)$, hence $\hat{v}(\phi)  = \hat{v}(\phi \wedge \psi)$. By \lna{CI}, we get $\hat{v}(\phi \tto (\phi \wedge \psi)) = \top$. Now use \lna{CK} to derive that $\hat{v}(\phi \tto \psi) = \top$. \neskip

For the second part, i.e., completeness: $\Th{\Mod{\cdot}}$ is clearly a closure operator, we only need to show that whenever $\Lambda \not\vdash \phi$,  there is $\gA \in \Mod{\Lambda}$ and a valuation $v$ s.t. $\hat{v}(\phi) \neq \top$. Simply pick $\gA$  to be the Lindenbaum-Tarski algebra of formulas quotiented by $\Lambda$-provable equivalence. We only need to show that this yields a $\tto$-algebra. The Heyting part is standard. For one half of \lna{CK}, we use \ka.  For the other, we use $\tr$ and $\na$. $\lna{CT}$ directly follows by $\tr$ and $\lna{CI}$ directly follows by \na.
\end{proof}


\section*{Proof of Lemma \ref{stabilogic}}
\begin{proof}
We choose $p$ distinct from ${\sf var}(\phi) \cup {\sf var}( \psi) \cup {\sf var}(\chi)$.
\neskip 

 We note that: 
{\footnotesize
\begin{multline*}
 \stex{p}{(((\phi \tto \psi) \wedge (\psi\tto \chi)) \to (\phi \tto \chi))} = \\
  ((((p \to \stex{p}{\phi}) \tto (p\to\stex{p}{\psi}))\; \wedge \\
  ((p\to\stex{p}{\psi})\tto (p\to \stex{p}{\chi}))) \; \to \\ 
  ((p\to\stex{p}{\phi}) \tto (p\to\stex{p}{\chi}))).
\end{multline*}
} 
  So the translation of an instance of {\sf Tr} is itself an instance of {\sf Tr}. It follows that
  \[
  \Lambda \vdash p \to  \stex{p}{(((\phi \tto \psi) \wedge (\psi\tto \chi)) \to (\phi \tto \chi))}.
  \]
 
  We note that:
 {\footnotesize 
\begin{multline*}
 \stex{p}{(((\phi \tto \psi) \wedge (\phi\tto \chi)) \to (\phi \tto (\psi \wedge \chi)))} = \\
  ((((p \to \stex{p}{\phi}) \tto (p\to\stex{p}{\psi}))\; \wedge \\
  ((p\to\stex{p}{\phi})\tto (p\to \stex{p}{\chi}))) \; \to \\ 
  ((p\to\stex{p}{\phi}) \tto (p\to\stex{p}{(\psi \wedge\chi)}))).
\end{multline*} \neskip 
}
  Over \iam, $(p\to\stex{p}{(\psi \wedge\chi)})$ is equivalent to
  \[
  (p\to \stex{p}{\psi}) \wedge (p\to \stex{p}{\chi}).
  \]
    So the translation of an instance of ${\sf K}_{\tt a}$  is, modulo \iam-provability, itself an instance of ${\sf K}_{\tt a}$. It follows that \[\Lambda \vdash p \to  \stex{p}{(((\phi \tto \psi) \wedge (\phi\tto \chi)) \to (\phi \tto (\psi\wedge\chi)))}.\]
\neskip

Suppose $\phi$ and $\phi \to \psi$ are in $\mathfrak S(\Lambda)$. Then, $\Lambda \vdash p \to \stex{p}{\phi}$ and $\Lambda \vdash p \to \stex{p}{(\phi\to \psi)}$. Since $\stex{p}{(\cdot)}$ commutes with the propositional connectives, we find $\Lambda \vdash p \to \stex{p}{\psi}$. \neskip

Suppose $(\phi \to \psi) \in \mathfrak S(\Lambda)$. Then we have $\Lambda \vdash p \to \stex{p}{(\phi\to \psi)}$. It follows that 
\[
\Lambda \vdash (p \to \stex{p}{\phi}) \to (p \to \stex{p}{\psi}). 
\]
Hence, $\Lambda \vdash (p \to \stex{p}{\phi}) \tto (p \to \stex{p}{\psi})$. Ergo, $\Lambda \vdash \stex{p}{(\phi\tto \psi)}$, and, \emph{a fortiori}, $\Lambda \vdash p \to \stex{p}{(\phi\tto \psi)}$. Hence, $(\phi \tto \psi)\in \mathfrak S(\Lambda)$.\neskip

Suppose $\phi\in \mathfrak S(\Lambda)$. We want to show that $\phi[q:=\psi]$ in $\Lambda$. We can arrange that $p$ is distinct from $q$. We now prove by induction on $p$-free formulas $\nu$ that 
\[
\stex{p}{(\nu[q:=\psi])} = (\stex{p}{\nu})[q := \stex{p}{\psi}].
\]
 Since  $\phi\in \mathfrak S(\Lambda)$, we have $\Lambda \vdash p \to \stex{p}{\phi}$. Hence,  
 \[
 \Lambda \vdash (p \to \stex{p}{\phi})[q := \stex{p}{\psi}],
 \]
  and so $\Lambda \vdash p \to (\stex{p}{\phi}[q := \stex{p}{\psi}])$. It follows that 
  \[
  \Lambda \vdash p \to \stex{p}{(\phi[q := \psi])}.
  \]
   Thus, $(\phi[q := \psi]) \in \mathfrak S(\Lambda)$.
\end{proof}


\section*{Proof of Lemma \ref{stabilinte}}

\begin{proof}
(i) is trivial. \neskip

We treat (ii). Suppose $p$ does not occur in $\phi$. Suppose $\phi \in \mathfrak S(\Lambda)$. Then, $\Lambda \vdash p \to\stex{p}{\phi}$. So, $\Lambda \vdash (p \to\stex{p}{\phi})[p:= \top]$. It is easy to see that $(p \to\stex{p}{\phi})[p:= \top]$ is equivalent to $\phi$ over \iam. \neskip

We treat (iii). By (ii) it is sufficient to show that $\mathfrak S\mathfrak S(\Lambda) \supseteq \mathfrak S(\Lambda)$. Let $p$ and $p'$ be distinct variables not in $\phi$. Suppose $\phi\in \mathfrak S(\Lambda)$. Then, $\Lambda \vdash p \to \stex{p}{\phi}$. It follows that  $\Lambda \vdash (p \to \stex{p}{\phi})[p := (p\wedge p')]$. We easily see that,  over \iam, $(p \to \stex{p}{\phi})[p := (p\wedge p')]$ is equivalent to $p'\to \stex{p'}{(p \to \stex{p}{\phi})}$.
\end{proof}


\section*{Proof of Theorem \ref{grotesmurf}}
\begin{proof}
We have $\opr\chi\top \vdash_{\iglam} \theta$. Hence, $\opr\chi\top \vdash_{\iglam} \dotbox (\top \iff \theta)$. So, it follows that $\opr\chi\top \vdash_{\iglam} \chi\top \iff \chi\theta$. Hence, by {\sloebra}, (a) $\vdash_{\iglam} \chi\top \ifff \chi\theta$. \neskip

We have $\opr \psi\opr\chi\top  \vdash_{\iglam} \opr\chi\top \iff \theta$. (We note that this step already works in \iam.)  So,  $\opr\psi\opr\chi\top \vdash_{\iglam} \psi\opr\chi\top \iff \psi\theta$. Thus, by {\sloebra}, we find (b): $ \vdash_{\iglam} \psi\opr\chi\top \ifff \psi\theta$. \neskip

Combining (a) and (b), we now have:
\qedright 
\begin{eqnarray*}
\iglam \vdash \theta & \iff & (\psi\opr\chi\top \tto \chi\top)\\
& \iff & (\psi\theta \tto \chi\theta) \\
& \iff & \phi\theta
\end{eqnarray*}
\end{proof}


\section*{Proof of Theorem \ref{wcircderiv}}

\begin{proof}
We work in $\iam \thad \loebb \thad \weakst$. \neskip

 Let $\beta := (\phi \wedge (\phi \tto \psi))$. We want to show that we have $(\beta \tto \psi) \to (\phi \tto \psi)$. Assume \neskip
 
 (a) $\opr((\beta \tto \psi) \to (\phi \tto \psi))$ \neskip
 
We apply {\weakst} with $\beta$ in the role of $\phi$ and $\psi$ in the role of $\psi$ obtaining: 
 \[ (\text{b})\;\; (\beta \tto \psi) \to (((\beta \tto \psi) \to \beta) \tto \psi).\]  
 By (a) we have:
 \begin{eqnarray*}
((\beta \tto \psi) \to \beta) & \ifff & ((\phi \tto \psi) \to (\phi \wedge (\phi\tto\psi))) \\
 & \ifff & ( (\phi \tto \psi) \to \phi) 
 \end{eqnarray*}
So (b) gives us: 
 \[ (\text{c})\;\;  (\beta \tto \psi) \to (((\phi \tto \psi) \to \phi) \tto \psi).\] 
 From (c) it is immediate that $(\beta \tto \psi) \to (\phi \tto \psi)$. \neskip 
 
 Finally, we apply L\"ob's Rule and we are done.
 \end{proof}


\section*{Proof of Theorem \ref{weakstderiv}}
\begin{proof}
We work in $\iam \thad\weako$. We will use that we have {\loebb} in this theory. Let $\alpha := (\phi \tto \psi) \to \phi$. We want to show that $(\phi \tto \psi) \to (\alpha \tto \psi)$. Assume \neskip 
 
 (a) $\opr((\phi \tto \psi) \to (\alpha \tto \psi))$ \neskip
 
 We apply \weako\ with $\alpha$ in the role of $\phi$ and $\psi$ in the role of $\psi$ obtaining: 
 \[ (\text{b})\;\; ((\alpha \wedge (\alpha \tto \psi))  \tto \psi) \to (\alpha \tto \psi).\]  
 By (a) we have:
 \begin{eqnarray*}
 (\alpha \wedge (\alpha \tto \psi)) & \ifff & (((\phi \tto \psi) \to \phi) \wedge (\phi \tto \psi)) \\
 & \ifff & (\phi \wedge (\phi \tto \psi)).
 \end{eqnarray*}
So (b) gives us: 
 \[ (\text{c})\;\; ((\phi \wedge (\phi \tto \psi))  \tto \psi) \to (\alpha \tto \psi).\] 
 From (c) it is immediate that $(\phi\tto \psi) \to (\alpha \tto \psi)$. \neskip
 
 Finally, we apply L\"ob's Rule and we are done.
\end{proof}


\section*{Proof of Theorem \ref{ququatao}}

\begin{proof}
We reason in $\ia\thad\quatao$. First, \quatao\ gives us that
\begin{eqnarray*}
\phi\tto (\psi \tto \chi) & \to & \phi \tto (\phi \tto (\psi \tto \chi)). \\
\end{eqnarray*}
On the other hand, it also yields that
\begin{eqnarray*}
\phi\tto (\psi \tto \chi) & \to & \phi \tto (\psi \tto (\psi \tto \chi)). \\
\end{eqnarray*}
Now use \di\ to derive
\begin{eqnarray*}
\phi\tto (\psi \tto \chi) & \to & \phi \tto ((\phi \vee \psi)  \tto (\psi \tto \chi)). \\
\end{eqnarray*}
\dots\ and now we  use $\quatao$ again:
\qedright
\begin{eqnarray*}
\phi\tto (\psi \tto \chi) & \to & \phi \tto ((\phi \vee \psi)  \tto (\psi \tto \chi)) \\
& \to & \phi \tto ((\phi \vee \psi)  \tto ((\phi \vee \psi) \tto (\psi \tto \chi))) \\
& \to & \phi \tto (\psi  \tto (\phi \tto (\psi \tto \chi))).
\end{eqnarray*}
\end{proof}

%% file: fourfour.tex
\section*{Countermodel used in the proof of Theorem \ref{iamququatao}}

The description follows the input format of Mace4 \cite{McCune03}. We ignore $\vee$. The connective $\tto$ is represented as $+$, whereas $\to$ is represented as $*$. The numbers correspond to indices used in the main text.

\begin{macecode}
formulas(assumptions).
x ^ (y ^ z) = (x ^ y) ^ z.
x ^ x = x.
x ^ y = y ^ x.
x ^ 1 = x.
x ^ 0 = 0.
x * x = 1.
x ^ (x * y) = x ^ y.
y ^ (x * y) = y.
x * (y ^ z) = (x * y) ^ (x * z).
(x + y) ^ (x + z) = x + (y ^ z).
((x + y) ^ (y + z)) ^ (x + z) = (x + y) ^ (y + z).
x + x = 1.
(x + y) ^ (x + (x + y)) = x + y.
end_of_list.

formulas(goals).
(x + (y + z)) ^ (x + (y + (x + (y + z)))) = x + (y + z).
end_of_list.

formulas(mace4_clauses).
x ^ (y ^ z) = (x ^ y) ^ z.
x ^ x = x.
x ^ y = y ^ x.
x ^ 1 = x.
x ^ 0 = 0.
x * x = 1.
x ^ (x * y) = x ^ y.
x ^ (y * x) = x.
x * (y ^ z) = (x * y) ^ (x * z).
(x + y) ^ (x + z) = x + (y ^ z).
((x + y) ^ (y + z)) ^ (x + z) = (x + y) ^ (y + z).
x + x = 1.
(x + y) ^ (x + (x + y)) = x + y.
(c1 + (c2 + c3)) ^ (c1 + (c2 + (c1 + (c2 + c3)))) != c1 + (c2 + c3).
end_of_list.

interpretation( 6, [number=1, seconds=8], [

        function(c1, [ 2 ]),

        function(c2, [ 4 ]),

        function(c3, [ 3 ]),

        function(*(_,_), [
			   1, 1, 1, 1, 1, 1,
			   0, 1, 2, 3, 4, 5,
			   3, 1, 1, 3, 4, 4,
			   2, 1, 2, 1, 1, 2,
			   0, 1, 2, 3, 1, 2,
			   3, 1, 1, 3, 1, 1 ]),

        function(+(_,_), [
			   1, 1, 1, 1, 1, 1,
			   0, 1, 0, 0, 0, 0,
			   5, 1, 1, 5, 2, 2,
			   0, 1, 0, 1, 1, 0,
			   0, 1, 0, 4, 1, 0,
			   4, 1, 1, 4, 1, 1 ]),

        function(^(_,_), [
			   0, 0, 0, 0, 0, 0,
			   0, 1, 2, 3, 4, 5,
			   0, 2, 2, 0, 5, 5,
			   0, 3, 0, 3, 3, 0,
			   0, 4, 5, 3, 4, 5,
			   0, 5, 5, 0, 5, 5 ])
]).

\end{macecode}